\documentclass[journal]{IEEEtran}
\usepackage{amsthm}
\makeatletter
\def\ps@IEEEtitlepagestyle{%
  \def\@oddfoot{\mycopyrightnotice}%
  \def\@evenfoot{}%
}
\def\mycopyrightnotice{%
  {\footnotesize This work has been submitted to the IEEE for possible publication. Copyright may be transferred without notice, after which this version may no longer be accessible.\hfill}
  \gdef\mycopyrightnotice{}
}

\usepackage[utf8]{inputenc}
\usepackage{graphicx}
\usepackage[center]{caption}
\graphicspath{{.}}
\usepackage{color}
\usepackage{epsfig}
\usepackage{latexsym}
\usepackage{amssymb}

\usepackage{mathrsfs}
\usepackage{bbm}
\usepackage{amsfonts, outlines, mathtools, outlines, cancel, hyperref, stackengine, bbm}
\usepackage{amsmath}

\usepackage{enumitem}

\setenumerate[1]{label=(\roman*)}

\usepackage{setspace}
\usepackage{subcaption}
\usepackage[ruled,vlined,linesnumbered]{algorithm2e}
\usepackage[english]{babel}
\usepackage[utf8]{inputenc}
\usepackage[T1]{fontenc}
\usepackage{newtxtext,newtxmath}
\usepackage{algorithmic}
\usepackage{scalerel}  

\newtheorem{theorem}{Theorem}

\newtheorem{assumption}{Assumption}
\newtheorem{lemma}{Lemma}
\newtheorem{remark}{Remark}

\newtheorem{envdef}{Definition}

\DeclareMathSizes{10}{9}{6}{5}
\allowdisplaybreaks

\newcommand{\inlnorm}[1]{\lVert#1\rVert}
\newcommand{\norm}[1]{\lVert#1\rVert}

\newcommand{\zer}[1]{\text{Zer}(#1)}
\newcommand{\fix}[1]{\text{Fix}(#1)}
\newcommand{\blkd}[1]{\text{blkd}(#1)}

\newcommand{\expt}[2]{\mathbb{E}_{#1}[#2]}
\newcommand{\bexpt}[2]{\mathbb{E}_{#1}\big[#2\big]}
\DeclarePairedDelimiter\ceil{\lceil}{\rceil} 

\newcommand\ubar[1]{\stackunder[1.2pt]{$#1$}{\rule{.8ex}{.075ex}}}
\newcommand{\gra}[1]{\text{gra}(#1)}
\DeclareMathOperator{\dom}{dom}

\newcommand{\supsub}[3]{{#1}^{#2}_{#3}}
\newcommand{\sumssub}[1]{\sum_{\scriptscriptstyle #1}}

\DeclareMathOperator{\minimize}{minimize}
\DeclareMathOperator{\maximize}{maximize}
\DeclareMathOperator{\subj}{subject\:to}

\newcommand{\playerN}{\mathcal{N}}
\newcommand{\edgeE}{\mathcal{E}}

\newcommand{\by}{\boldsymbol y}
\newcommand{\bz}{\boldsymbol z}

\newcommand{\bone}{\boldsymbol 1}
\newcommand{\bzero}{\boldsymbol 0}
\newcommand{\btau}{\boldsymbol \tau}
\newcommand{\idty}{\text{Id}}

\newcommand{\rset}[2]{\mathbb{R}^{#1}_{#2}}
\newcommand{\pdset}[2]{\mathbb{S}^{#1}_{#2}}
\newcommand{\crset}[2]{\overline{\mathbb{R}}^{#1}_{#2}}
\newcommand{\nset}[2]{\mathbb{N}^{#1}_{#2}}

\newcommand{\neighbN}[2]{\mathcal{N}^{#1}_{#2}}

\DeclareMathOperator{\argmin}{argmin}
\DeclareMathOperator{\argmax}{argmax}

\usepackage{accents}

\DeclareMathOperator{\poptA}{\bar{\mathbb{A}}}
\DeclareMathOperator{\appoptA}{\bar{\mathcal{A}}}
\DeclareMathOperator{\poptB}{\bar{\mathcal{B}}}

\newcommand{\bmu}{\boldsymbol \mu}
\newcommand{\blambda}{\boldsymbol \lambda}
\newcommand{\optAy}{\mathbb{A}_y}
\newcommand{\diagA}{\Lambda}

\newcommand{\pspace}{\mathcal{K}}
\DeclareMathOperator{\proj}{Pj}
\newcommand{\res}[1]{\text{res}(#1)}

\newcommand{\rrcomp}{\mathscr R}

\DeclareMathOperator{\optT}{\mathbb{T}}
\DeclareMathOperator{\optA}{\mathbb{A}}
\DeclareMathOperator{\optB}{\mathcal{B}}

\DeclareMathOperator{\gjacob}{\mathbb{F}}
\DeclareMathOperator{\extgjacob}{\tilde{\mathbb{F}}}

\newcommand{\Tblue}[1]{\textcolor{black}{#1}}
\usepackage{cite}
\usepackage{textcomp}
\def\BibTeX{{\rm B\kern-.05em{\sc i\kern-.025em b}\kern-.08em
    T\kern-.1667em\lower.7ex\hbox{E}\kern-.125emX}}

\begin{document}

\title{Distributed Computation of Stochastic GNE with Partial Information: An Augmented Best-Response Approach}
\author{Yuanhanqing Huang and Jianghai Hu
\thanks{This work was supported by the National Science Foundation under Grant No. 2014816. A conference version has been submitted to ACC 2022 \cite{huang2021sgnep}.}
\thanks{The authors are with the Elmore Family School of Electrical and Computer Engineering, Purdue University, West Lafayette, IN, 47907, USA (e-mail: huan1282@purdue.edu; jianghai@purdue.edu).}
}

\maketitle

\begin{abstract}
In this paper, we focus on the stochastic generalized Nash equilibrium problem (SGNEP) which is an important and widely-used model in many different fields. 
In this model, subject to certain global resource constraints, a set of self-interested players aim to optimize their local objectives that depend on their own decisions and the decisions of others and are influenced by some random factors. 
We propose a distributed stochastic generalized Nash equilibrium seeking algorithm in a partial-decision information setting based on the Douglas-Rachford operator splitting scheme, \Tblue{which relaxes assumptions in the existing literature}.
The proposed algorithm updates players' local decisions through augmented best-response schemes and subsequent projections onto the local feasible sets, which occupy most of the computational workload. 
The projected stochastic subgradient method is applied to provide approximate solutions to the augmented best-response subproblems for each player. 
The Robbins-Siegmund theorem is leveraged to establish the main convergence results to a true Nash equilibrium using the proposed inexact solver. 
Finally, we illustrate the validity of the proposed algorithm \Tblue{via} two numerical examples, i.e., a stochastic Nash-Cournot distribution game and a multi-product assembly problem with the two-stage model.
\end{abstract}


\section{Introduction}


In the Nash equilibrium problem (NEP), a set of self-interested players aim to optimize their individual payoffs which depend not only on their own decisions but also on the decisions of others \cite{nash1950equilibrium}. 
The generalized Nash equilibrium problem (GNEP) extends the NEP by considering additional global resource constraints that these players should collectively satisfy \cite{facchinei2007generalized,facchinei2010generalized}. 
In numerous practical applications, such as strategic behaviors in power markets \cite{kannan2011strategic,kannan2013addressing}, engagement of multiple humanitarian organizations in disaster relief \cite{nagurney2020stochastic}, and the traffic assignment of strategic risk-averse users \cite{nikolova2014mean}, in addition to couplings in objectives and global constraints, there is also uncertainty involved in the objective functions. 
For example, in the target-rate power management problem for wireless networks, the power of battery-driven devices should be regulated in a real-time manner in the presence of inherent stochastic fluctuations of the underlying network \cite{zhou2021robust}.
These applications motivate us to consider an extension to the GNEP, called the stochastic generalized Nash equilibrium problem (SGNEP). 
In the SGNEP, instead of deterministic objective functions, players optimize the expected values of uncertain objective functions which are dependent on some random variables. 
Although the SGNEP can capture a wider range of applications, computing its Nash equilibria becomes a much more challenging problem compared to the GNEP, due to the absence of closed-form expressions of the objective functions. 
Fortunately, as has been shown in \cite[Sec.~1.4]{facchinei2007finite}, many SGNEPs can be formulated as stochastic variational inequalities (SVIs) and solved by leveraging existing results from this field, e.g. \cite{koshal2012regularized, bot2020mini, yousefian2017smoothing, cui2016analysis, kannan2019optimal}. 

Our aim in this paper is to develop a distributed algorithm under the partial-decision information setting for solving SGNEPs over a network of players. 
In the context of non-cooperative games on networks, participants are self-interested and make rational decisions that maximize their own payoffs. 
It is unrealistic that these participants will adopt any centralized methods that require cooperation among them. 
Because of this, there is an enduring research interest in distributing the computation of Nash equilibria \cite{salehisadaghiani2016distributed, parise2020distributed}, especially through the avenue of operator splitting technique \cite{yi2019operator, yi2018distributed}. 
In addition to the distributed computation, under most circumstances, participants can only have access to local information and decisions of their neighbors, which constitutes a partial-decision information setting \cite{pavel2019distributed, bianchi2022fast, belgioioso2020distributed}. 
One reason is that these players are reluctant to send their local information and decisions to the general public out of privacy concerns while being willing to share these with their trusted neighbors on the whole network. 
Although there may exist some central node that collects and distributes the necessary information to each player, this central node is subject to performance limitations, such as single point of failure, and limited flexibility and scalability \cite{yang2019survey}. 
The limited capacity of communication channels also constraints information sharing over the network. 


Significant efforts have been devoted to designing algorithms to solve SGNEPs distributedly under the full-decision information setting where each player has access to all other players' decisions. 
The authors of \cite{koshal2012regularized} consider a Cartesian stochastic variational inequality problem with a monotone map. 
They propose a stochastic iterative Tikhonov regularization method and a stochastic iterative proximal-point method to solve it, which let players update their regularization parameters and centering term properly after each iteration. 
The authors of \cite{franci2020distributed} propose a solution based on the preconditioned forward-backward (FB) operator splitting with the expected-value pseudogradient assumed to be restricted co-coercive and approximated via the stochastic approximation (SA) scheme. 
To accelerate game dynamics and relax the co-coercivity assumption, \cite{cui2021relaxed} adopts a forward-backward-forward framework. 
These works are mostly gradient-based which enjoy low complexity in implementation and updating per player step. 
Nevertheless, rational players would take best-response actions given others' strategies and deviate from gradient-based schemes unless forced by some external authority. 
The work in \cite{lei2020synchronous} provides an inexact generalization of the proximal best-response (BR) schemes to the SNEP whose corresponding proximal BR map admits a contractive property.  
\Tblue{The authors of \cite{lei2018distributed} further consider SNEPs with composite objectives and design a variable sample-size proximal BR scheme, under a contractive property on the proximal BR map.}
Yet, far less has been studied when it comes to the distributed solution to SGNEPs with merely partial information. 
The only existing work to our best knowledge is \cite{franci2020stochastic}, which also relies on the FB framework along with the SA method. 
The convergence of the proposed algorithms has been analyzed under the assumption that the preconditioned forward operator is restricted co-coercive, which only allows comparatively small step sizes. 

Our contributions can be summarized in the following aspects. First, we propose a distributed algorithm to solve the SGNEP under the partial-decision information setting based on the Douglas-Rachford splitting and the proximal mapping. 
In the proposed algorithm, the involved players are asked to update their decision vectors in two separate steps: solving the augmented best-response subproblems, and projecting onto the local feasible sets after some linear transformations. 
The updates of their local estimates and dual variables only require some trivial linear transformations. 
This algorithm can deal with cases where the scenario-based objectives of players are nonsmooth, \Tblue{and relaxes some commonly-made assumptions such as the $\alpha-$cocoercivity with $\alpha>1/2$ in the FB splitting}. 
Second, we establish the convergence of the proposed algorithm under assumptions concerning the properties of the pseudogradient, the extended pseudogradient, and the stochastic subgradients of the objectives. 
Without relying on the contractive property, the proof in this paper is based on the Robbins-Siegmund theorem and extends the convergence results discussed in \cite{lei2020synchronous}. 
Drawing tools and techniques from stochastic approximation and convex analysis, we then construct a feasible inexact solver for the augmented best-response subproblems based on the projected stochastic subgradient method and discuss the prescribed accuracy within which the inexact solver should achieve such that the algorithm convergence is ensured. 
The proposed inexact solver based on the projected stochastic subgradient method requires the projection onto some bounded box sets rather than the (potentially complicated) local feasible sets of the original problem, which considerably improves the computational efficiency. 

The remainder of this paper is organized as follows. 
In Sec.~\ref{sect:prob-form}, we formally formulate the SGNEP on networks and provide some basic definitions as well as assumptions. 
\Tblue{In this section, we recast the SGNEP as the zero-finding problem of a certain operator and prove that the solution of the latter is a "variational" solution of the former. }
In Sec.~\ref{sect:dr-algm}, a distributed algorithm in a partial-decision information setting is proposed. 
Sec.~\ref{sect:convg} focuses on the convergence analysis of the proposed algorithm. 
In this section, we also construct an inexact solver based on the projected stochastic subgradient method. 
In Sec.~\ref{sect:simu}, to demonstrate the theoretical findings and the proposed algorithm in practical applications, we include two numerical examples: a stochastic Nash-Cournot distribution game and a multi-product assembly problem with the two-stage model.  
Sec.~\ref{sect:conclu} concludes the paper and highlights potential extensions and applications. 

\textit{Basic Notations:} For a set of matrices $\{V_i\}_{i \in S}$, we let $\blkd{V_1, \ldots, V_{|S|}}$ or $\blkd{V_i}_{i \in S}$ denote the diagonal concatenation of these matrices, $[V_1, \ldots, V_{|S|}]$ their horizontal stack, and $[V_1; \cdots; V_{|S|}]$ their vertical stack. 
For a set of vectors $\{v_i\}_{i \in S}$, $[v_i]_{i \in S}$ or $[v_1; \cdots; v_{|S|}]$ denotes their vertical stack. 
For a matrix $V$ and a pair of positive integers $(i, j)$, $[V]_{(i,j)}$ denotes the entry on the $i$-th row and the $j$-th column of $V$. 
For a vector $v$ and a positive integer $i$, $[v]_i$ denotes the $i$-th entry of $v$. 
Denote $\crset{}{} \coloneqq \rset{}{} \cup \{+\infty\}$, $\rset{}{+} \coloneqq [0, +\infty)$, and $\rset{}{++} \coloneqq (0, +\infty)$. 
$\pdset{n}{+}$ (resp. $S^n_{++}$) represents the set of all $n\times n$ symmetric positive semi-definite (resp. definite) matrices.
$\iota_{\mathcal{S}}(x)$ is defined to be the indicator function of a set $\mathcal{S}$, i.e., if $x \in \mathcal{S}$, then $\iota_{\mathcal{S}}(x) = 0$; otherwise, $\iota_{\mathcal{S}}(x) = +\infty$. 
$N_{S}(x)$ denotes the normal cone to the set $S \subseteq \rset{n}{}$ at the point $x$: if $x \in S$, then $N_S(x) \coloneqq \{u \in \rset{n}{} \mid \sup_{z \in S} \langle u, z-x \rangle \leq 0 \}$; otherwise, $N_S(x) \coloneqq \varnothing$. 
If $S \in \rset{n}{}$ is a closed and convex set, the map $\proj_S:\rset{n}{} \to S$ denotes the projection onto $S$, i.e., $\proj_S(x) \coloneqq \argmin_{v \in S} \norm{v - x}_2$. 
We use $\rightrightarrows$ to indicate a point-to-set map. 
For an operator $T: \rset{n}{} \rightrightarrows \rset{n}{}$, $\zer{T} \coloneqq \{x \in \rset{n}{} \mid Tx \ni \bzero\}$ and $\fix{T} \coloneqq \{x \in \rset{n}{} \mid Tx \ni x\}$ denote its zero set and fixed point set, respectively. 
We denote $\dom(T)$ the domain of the operator $T$ and $\gra{T}$ the graph of it.
The resolvent and reflected resolvent of $T$ are defined as $J_T \coloneqq (I + T)^{-1}$ and $R_T \coloneqq 2J_T - I$, respectively.

\section{Problem Formulation}\label{sect:prob-form}

\subsection{Stochastic Game Formulation and SGNE}

In this section, we formulate the stochastic generalized Nash equilibrium problem (SGNEP) on networks. There are $N$ players participating in the game, indexed by $\playerN \coloneqq \{1, \ldots, N\}$. 
Each player $i \in \playerN$ needs to determine its local decision vector $x_i \in \mathcal{X}_i$ to optimize its objective, where $\mathcal{X}_i \subseteq \rset{n_i}{}$ is the local feasible set/action space of player $i$. 
This Nash equilibrium seeking problem is generalized because, besides the local constraints $\{\mathcal{X}_i\}_{i \in \playerN}$, the decision vectors of all the players should satisfy some global resource constraints, i.e., $\sum_{i \in \playerN}A_ix_i \leq c$. 
Here, we have the matrix $A_i \in \rset{m \times n_i}{}$ with $m$ denoting the number of the (global) affine coupling constraints, and the constant vector $c \in \rset{m}{}$ representing the quantities of available resources. 
Altogether, for each player $i$, the feasible set of the decision vector $x_i$ is given by 
\begin{equation}
    \tilde{\mathcal{X}}_i(x_{-i}) \coloneqq \mathcal{X}_i \cap \{x_i \in \rset{n_i}{} \mid A_ix_i + {\textstyle\sum}_{j \in \playerN_{-i}}A_jx_j \leq c\},
\end{equation}
where $\playerN_{-i} \coloneqq \playerN \backslash \{i\}$, and $x_{-i}$ denotes the stack of decision vectors except that of player $i$. 
Accordingly, the feasible set of the collective decision vector $x \coloneqq [x_1; \cdots; x_N]$ is given by 
\begin{equation}
    \tilde{\mathcal X} \coloneqq \mathcal X \cap \{x \in \rset{n}{} \mid Ax - c \leq \bzero\},
\end{equation}
where $\mathcal{X} \coloneqq \prod_{i \in \playerN} \mathcal{X}_i$, $n \coloneqq \sum_{i \in \playerN} n_i$, and $A \coloneqq [A_1, A_2, \ldots, A_N]$.

To capture uncertainty in practical settings, we consider stochastic games where the objective function $\mathbb{J}_i(x_i; x_{-i})$ of each player $i$ is the expected value of certain function $J_i$.
Formally, given the decisions $x_{-i}$ of the other players, each player $i$ aims to solve the following local problem:
\begin{equation}\label{eq:pf-optprob}
\begin{cases}
\minimize_{x_i \in \mathcal{X}_i} \mathbb{J}_i(x_i; x_{-i}) = \expt{\xi_i}{J_i(x_i; x_{-i}, \xi_{i})} \\
\subj \qquad A_ix_i \leq c - \sum_{j \in \playerN_{-i}} A_jx_j
\end{cases},
\end{equation}
where $\xi_i: \Omega_i \to \rset{n_{\xi_i}}{}$ is a random variable in a well-defined probability space.

Given the above formulation of the SGNEP, we have the following standing assumptions that hold throughout the paper. 
\begin{assumption}\label{asp:subgrad}
(Scenario-Based Objectives) For each player $i \in \playerN$, given any fixed sample $\omega_i \in \Omega_i$, the scenario-based objective $J_i(\cdot; \cdot, \xi_i(\omega_i))$ is proper and continuous. 
In addition, $J_i(x_i; x_{-i}, \xi_i(\omega_i))$ is a convex function w.r.t. $x_i$ given any fixed $x_{-i}$ and $\omega_i \in \Omega_i$.
\end{assumption}

\begin{assumption}\label{asp:fesb-set}
(Feasible Sets) Each local feasible set $\mathcal{X}_i$ is nonempty, compact, and convex. 
The collective feasible set $\tilde{\mathcal{X}}$ is nonempty, and the Mangasarian-Fromovitz constraint qualification (MFCQ) holds \cite[Ch~3.2]{facchinei2007finite}\cite[Sec.~16.2.3]{palomar2010convex}. 
\end{assumption}

Collectively solving the problems in \eqref{eq:pf-optprob} gives rise to the stochastic generalized Nash equilibrium (SGNE), the formal definition of which is given as follows \cite{franci2020distributed, ravat2011characterization}:
\begin{envdef}\label{def:sgne}
The collective decision $x^* \in \tilde{\mathcal{X}}$ is a stochastic generalized Nash equilibrium (SGNE) if no player can benefit by unilaterally deviating from $x^*$. Specifically, for all $i \in \playerN$, $\mathbb{J}_i(x^*_i; x^*_{-i}) \leq \mathbb{J}_i(x_i; x^*_{-i})$ for any $x_i \in \tilde{X}_i(x^*_{-i})$.
\end{envdef}


Under Assumption~\ref{asp:subgrad}, the SGNE seeking problems can be transformed to the corresponding generalized quasi-variational inequality (GQVI) \cite[Sec.~12.2]{palomar2010convex}.
As shown in \cite[Prop.~12.3]{palomar2010convex}, Definition~\ref{def:sgne} of SGNE coincides with the following definition from the perspective of GQVI: 
\begin{envdef}\label{def:sgne2}
The collective decision $x^* \in \tilde{\mathcal{X}}$ is a stochastic generalized Nash equilibrium (SGNE) if $x^*$ along with a suitable $g^* \in \prod_{i \in \playerN}\partial_{x_i}\mathbb{J}_i(x^*_i; x^*_{-i})$ is a solution of the problem  $\text{GQVI}(\prod_{i \in \playerN}\tilde{\mathcal{X}}_i, \prod_{i \in \playerN}\partial_{x_i}\mathbb{J}_i)$, i.e., 
\begin{equation}
\begin{split}
(x - x^*)^Tg^* \geq 0, \forall x \in {\textstyle\prod}_{i \in \playerN}\tilde{\mathcal{X}}_i(x^*_{-i}).
\end{split}
\end{equation}
\end{envdef}

As suggested in \cite[Sec.~12.2.3]{palomar2010convex}, under Assumptions~\ref{asp:subgrad} and \ref{asp:fesb-set}, we can equivalently recast the problem in \eqref{eq:pf-optprob} into a set of inclusions by considering the Karush-Kuhn-Tucker (KKT) conditions of the above GQVI such that $\forall i \in \playerN$:
\begin{equation}\label{def:kkt-gqvi}
\begin{split}
    & \bzero \in \partial_{x_i}\mathbb{J}_i(x^*_i; x^*_{-i}) + A_i^T\lambda_i + N_{\mathcal{X}_i}(x^*_i) \\
    & \bzero \in -(Ax^* - c) + N_{\rset{m}{+}}(\lambda_i),
\end{split}
\end{equation}
where $\lambda_i$ is the Lagrangian multiplier for the global resource constraints $A_ix_i \leq c - \sum_{j \in \neighbN{}{-i}}A_jx_j$ for each player $i$. 

In this paper, we restrict our attention to a subset of these SGNEs where the players share the same coupled constraints, and hence all the Lagrangian multipliers are in consensus, i.e., $\lambda_1 = \ldots = \lambda_N$. 
This gives rise to a generalized variational inequality (GVI) problem. 
This subclass of the SGNEs, known as the variational stochastic generalized Nash equilibria (v-SGNEs) \cite{facchinei2010generalized, facchinei2007generalized}, enforces the idea of economic fairness and enjoys better social stability/sensitivity \cite{kulkarni2012variational}. 
We will focus on this subclass since we can leverage a variety of tools that have been developed for solving (G)VIs \cite[Ch.~10-12]{facchinei2007finite} and design the modified best-response dynamics based on it. 
\begin{envdef}\label{def:v-sgne}
The collective decision $x^* \in \tilde{\mathcal{X}}$ is a variational stochastic generalized Nash equilibrium (v-SGNE) if $x^*$ along with a suitable $g^* \in \prod_{i \in \playerN}\partial_{x_i}\mathbb{J}_i(x^*_i; x^*_{-i})$ is a solution of the $\text{GVI}(\tilde{\mathcal X}, \prod_{i \in \playerN}\partial_{x_i}\mathbb{J}_i)$, i.e., 
\begin{equation}\label{eq:v-sgne}
\begin{split}
   (x - x^*)^Tg^* \geq 0, \forall x \in \tilde{\mathcal X}. 
\end{split}
\end{equation}
\end{envdef}
Similarly, the KKT system of the above GVI is given by:
\begin{equation}\label{eq:kkt-gvi}
\begin{split}
    & \bzero \in \partial_{x_i}\mathbb{J}_i(x^*_i; x^*_{-i}) + A_i^T\lambda + N_{\mathcal{X}_i}(x^*_i) \\
    & \bzero \in -(Ax^* - c) + N_{\rset{m}{+}}(\lambda),
\end{split}
\end{equation}
where $\lambda$ is the Lagrangian multiplier for the global constraints in \eqref{eq:pf-optprob}. 
Notice that the GVI in \eqref{eq:v-sgne} is not completely equivalent to the initial SGNEP in \eqref{eq:pf-optprob} as the game may admit an SGNE while the GVI has no solution. 
We make the following assumption concerning the existence of v-SGNEs. 
\begin{assumption}\label{asp:ne-exist}
(Existence of v-SGNE) The SGNEP considered admits a nonempty set of v-SGNEs. 
\end{assumption}

\begin{remark}
The existence and multiplicity of solutions of GNEPs with continuously differentiable objectives have been extensively studied, and the related theories can be found in \cite[Ch.~2\&3]{facchinei2007finite}. 
For the GNEPs with nonsmooth objectives, we can check the existence of v-GNEs of these GNEPs by \cite[Prop.~12.11]{palomar2010convex}. 
If the closed-form expressions of the objectives $\mathbb{J}_i(x_i; x_{-i})$ for any $i \in \playerN$ are unavailable and we cannot apply the above results,  \cite[Sec.~4]{ravat2011characterization} provides sufficient conditions to guarantee the existence of v-SGNEs based on the properties of scenario-based objectives. 
\end{remark}

\subsection{Network Game Formulation}

In network games, there exists an underlying communication graph $\mathcal{G} = (\playerN_g, \edgeE_g)$, where players can communicate with their neighbors through arbitrators on the edges. 
The node set $\playerN_g$ denotes the set of all players, and $\edgeE_g \subseteq \playerN_g \times \playerN_g$ is the set of directed edges. The cardinalities $|\playerN_g|$ and $|\edgeE_g|$ are denoted by $N_g$ and $E_g$. In this case, $\playerN_g = \playerN$ and $N_g = N$. We use $(i, j)$ to denote a directed edge having node/player $i$ as its tail and node/player $j$ as its head. 
For notational brevity, let $\playerN_i$ denote the set of immediate neighbors of player $i$ who can directly communicate with it, $\neighbN{+}{i} \coloneqq \{j\in\mathcal{N}\mid(j, i)\in\edgeE_g\}$ the set of in-neighbors of player $i$, and $\neighbN{-}{i} \coloneqq \{j\in\mathcal{N}\mid(i, j)\in\edgeE_g \}$ the set of out-neighbors of player $i$. 
Note that although the multipliers we are going to introduce are defined in a directed fashion, we assume each node can send messages to both its in- and out-neighbors, and $\mathcal{G}$ should satisfy the following assumption. 
\begin{assumption}{(Communicability)}\label{asp:commtopo}
The underlying communication graph $\mathcal{G} = (\mathcal{N}_g, \edgeE_g)$ is undirected and connected. Besides, it has no self-loops. 
\end{assumption}

We next recast the SGNEP in \eqref{eq:pf-optprob} as the zero-finding problem of a certain operator that can be carried out distributedly over the communication graph $\mathcal{G}$ via the network Lagrangian of this game and refer the interested reader to \cite{huang2021distributed} for more details. 
Now for each player $i \in \playerN$, besides its local decision vector $y^i_i \in \mathcal{X}_i$, it keeps a local estimate $y^j_i \in \rset{n_j}{}$ of the player $j$'s decision for all $j \in \neighbN{}{-i}$, which together constitutes its augmented decision vector $y_i$. 
Here, we use $y^i_i$ to denote the local decision of each player $i$ to distinguish from the case where only local decision $x_i$ are maintained and considered. 
We denote $y^{-i}_i \coloneqq [y^j_i]_{j \in \neighbN{}{-i}}$ the vertical stack of $\{y^j_i\}_{j \in \neighbN{}{-i}}$ and $y_i \coloneqq [y^j_i]_{j \in \playerN}$ the vertical stack of $\{y^j_i\}_{j \in \playerN}$, both in prespecified orders. 
Denote $n_{<i} = \sum_{j\in\mathcal{N}, j<i} n_j$ and $n_{>i} = \sum_{j\in\mathcal{N}, j>i}n_j$. The extended feasible set of $\by \coloneqq [y_i]_{i \in \playerN}$ is defined as $\hat{\mathcal{X}} \coloneqq \hat{\mathcal{X}}_1 \times \hat{\mathcal{X}}_2 \times \cdots \times \hat{\mathcal{X}}_N$ with each one defined as $\hat{\mathcal{X}}_i \coloneqq \mathbb{R}^{n_{<i}} \times \mathcal{X}_i \times \mathbb{R}^{n_{>i}}$. 
For brevity, we shall write $\{y_i\}$ in replacement of the more cumbersome notation $\{y_i\}_{i \in \playerN}$ and similarly for other variables on nodes and edges (e.g. the dual variables $\{\mu_{ji}\}_{(j,i) \in \edgeE_g}$ to be introduced below will be denoted simply by $\{\mu_{ji}\}$), unless otherwise specified. 
In the reformulated zero-finding problem, we introduced a set of dual variables $\{\lambda_i\}$ to enforce the global resource constraints. 
Moreover, another two sets of dual variables $\{\mu_{ji}\}$ and $\{z_{ji}\}$ are introduced to guarantee the consensus of $\{y_i\}$ and $\{\lambda_i\}$. 
It is worth mentioning that $\{y_i\}$ and $\{\lambda_i\}$ are maintained by players while $\{\mu_{ji}\}$ and $\{z_{ji}\}$ are maintained by arbitrators on the edges. 

We next give a brief introduction to two commonly used operators in the distributed solution of GNEP: the pseudogradient $\gjacob: \rset{n}{} \rightrightarrows \rset{n}{}$ and the extended pseudogradient $\extgjacob: \rset{nN}{} \rightrightarrows \rset{n}{}$. The pseudogradient $\gjacob$ is the vertical stack of the partial subgradients of the objective functions of all players, which is given as follows:
\begin{equation}\label{eq:psd-jacob}
    \gjacob: x \mapsto [\partial_{x_i}\mathbb{J}_i(x_i; x_{-i})]_{i \in \playerN}.
\end{equation}
In contrast, the extended pseudogradient $\extgjacob$ defined in \eqref{eq:ext-psd-jacob} is a commonly used operator under the partial-decision information setting, where each player keeps the local estimates of others' decisions and then uses these estimates as the parametric inputs: 
\begin{equation}\label{eq:ext-psd-jacob}
    \extgjacob: [y_i]_{i \in \playerN} \mapsto [\partial_{y^i_i} \mathbb{J}_i(y^i_i; y^{-i}_i)]_{i \in \playerN}.
\end{equation}
To incorporate the extended pseudogradient $\extgjacob$ into a fixed-point iteration, we then introduce the individual selection matrices $\{\mathcal{R}_i\}_{i \in \playerN}$ and their diagonal concatenation $\mathcal{R} \in \rset{n \times nN}{}$:
\begin{equation}\label{eq:prjr}
\begin{split}
    & \mathcal{R}_i = [\bzero_{n_i\times n_{<i}}, \mathbf{I}_{n_i}, \bzero_{n_i\times n_{>i}}], \;
    \mathcal{R} = \blkd{\mathcal{R}_1, \ldots, \mathcal{R}_N}.
\end{split}
\end{equation}
Notice that $y^i_i = \mathcal{R}_iy_i$ and $\mathcal{R}_i\mathcal{R}_i^T = I_{n_i}$. 
Finally, the set-valued operator $\optT$ we are going to study is given below: 
\begin{equation}\label{eq:optT}
\small
    \optT: 
    \begin{bmatrix}\by\\ \blambda\\ \bmu\\ \bz\end{bmatrix} \mapsto 
    \begin{bmatrix}
    \mathcal{R}^T(\extgjacob(\by) + \diagA^T\blambda)+ B_n\bmu + \rho_\mu L_n\by + N_{\hat{\mathcal{X}}}(\by) \\
    N_{\rset{mN}{+}}(\blambda) - \diagA \mathcal{R} \by + \boldsymbol c + B_m\bz + \rho_z L_m\blambda \\
    -B_n^T\cdot\by \\
    -B_m^T\cdot\blambda
    \end{bmatrix},
\normalsize
\end{equation}
where $\diagA$ is the diagonal concatenation of $\{A_i\}_{i \in \playerN}$, i.e., $\diagA \coloneqq \blkd{A_1, \ldots, A_N}$; 
$\boldsymbol c$ is the vertical stack of $\{c_i\}_{i \in \playerN}$ with $\sum_{i \in \playerN}c_i = c$; 
$B_n \coloneqq (B \otimes I_n)$, $L_n \coloneqq (L \otimes I_n)$, $B_m \coloneqq (B \otimes I_m)$, $L_m \coloneqq (L \otimes I_m)$, $B$ and $L$ are the incidence matrix and Laplacian matrix of the underlying communication graph, respectively, with $L = B\cdot B^T$; 
and $\by$, $\blambda$, $\bmu$, and $\bz$ are the stack vectors of $\{y_i\}$, $\{\lambda_i\}$, $\{\mu_{ji}\}$, and $\{z_{ji}\}$, respectively; $\psi$ denotes the stack of the primal and dual variables, i.e., $\psi \coloneqq [\by; \blambda; \bmu; \bz]$.

\begin{theorem}\label{thm:zerokkt}
Suppose Assumptions \ref{asp:subgrad} to \ref{asp:commtopo} hold, and there exists $\psi^* \coloneqq [\by^*; \blambda^*; \bmu^*; \bz^*] \in \zer{\optT}$. Then $\by^* = \bone_N \otimes y^*$, $\blambda^* = \bone_N \otimes \lambda^*$, and $(y^*, \lambda^*)$ satisfies the KKT conditions \eqref{eq:kkt-gvi} for v-GNE with $x^*$ replaced with $y^*$. 
Conversely, for a solution $(y^\dagger, \lambda^\dagger)$ of the KKT problem in \eqref{eq:kkt-gvi}, there exist $\bmu^\dagger$ and $\bz^\dagger$ such that $\psi^\dagger \coloneqq [\bone_N \otimes y^\dagger; \bone_N \otimes\lambda^\dagger; \bmu^\dagger; \bz^\dagger] \in \zer{\optT}$.
\end{theorem}
\begin{proof}
See the proof of \cite[Thm.~1]{huang2021distributed}.
\end{proof}

Thus, finding a v-SGNE of the game in \eqref{eq:pf-optprob} is equivalent to solving for a zero point of the operator $\mathbb{T}$. 
To facilitate the convergence analysis of the algorithm to be proposed for the latter task, we make two parallel assumptions, either of which is instrumental for the convergence proof in Sect.~\ref{sect:convg}. 
\begin{assumption}{(Convergence Condition)}\label{asp:convg}
At least one of the following statements holds:
\begin{outline}[enumerate]
\1 the operator $\mathcal{R}^T\extgjacob+ \frac{\rho_\mu}{2}L_n$ is maximally monotone;
\1 the pseudogradient $\gjacob$ is strongly monotone and Lipschitz continuous, i.e., there exist $\eta > 0$ and $\theta_1 > 0$, such that $\forall x, x' \in \rset{n}{}$, $\langle x - x', \gjacob(x) - \gjacob(x')\rangle \geq \eta \norm{x - x'}^2$ and $\norm{\gjacob(x) - \gjacob(x')} \leq \theta_1\norm{x - x'}$. The operator $\mathcal{R}^T\extgjacob$ is Lipschitz continuous, i.e., there exists $\theta_2 > 0$, such that $\forall \by, \by' \in \rset{nN}{}$, $\norm{\extgjacob(\by) - \extgjacob(\by')} \leq \theta_2 \norm{\by - \by'}$. 
\Tblue{Moreover, $\rho_\mu \geq \frac{2}{\sigma_1}(\frac{(\theta_1 + \theta_2)^2}{4\eta} + \theta_2)$, where $\sigma_1$ is the smallest positive eigenvalue of $L$.}
\end{outline}
\end{assumption}

\section{An Augmented Best-Response Scheme}\label{sect:dr-algm}

To compute the zeros of the operator $\optT$ given in the previous section, we leverage the Douglas-Rachford (DR) splitting method which combines operator splitting and the \Tblue{Krasnoselskii-Mann} (K-M) schemes. 
Given a nonexpansive operator $Q$ with a nonempty fixed point set $\fix{Q}$, the K-M scheme \cite[Sec.~5.2]{BauschkeHeinzH2017CAaM} suggests the following iteration:
\begin{equation}\label{eq:KM-iter}
    \psi^{(k+1)} \coloneqq \psi^{(k)} + \gamma^{(k)}(Q\psi^{(k)} - \psi^{(k)}),
\end{equation}
where $(\gamma^{(k)})_{k \in \nset{}{}}$ is a sequence such that $\gamma^{(k)} \in [0, 1]$ for all $k \in \nset{}{}$ and $\sum_{k \in \nset{}{}}\gamma^{(k)}(1 - \gamma^{(k)}) = \infty$. 
Here, we introduce a set of local bounded box constraints $\{\mathcal{X}^B_i\}$ which can be chosen manually as long as it satisfies $\mathcal{X}_i \subseteq \mathcal{X}^B_i$ for all $i \in \playerN$.
We similarly define the extended box set $\hat{\mathcal{X}}^B \coloneqq \hat{\mathcal{X}}^B_1 \times \hat{\mathcal{X}}^B_2 \times \cdots \hat{\mathcal{X}}^B_N$ where the extended box set of each player $i$ is defined as $\hat{\mathcal{X}}^B_i \coloneqq \mathbb{R}^{n_{<i}} \times \mathcal{X}^B_i \times \mathbb{R}^{n_{>i}}$. 
It is easy to see that the normal cones of $\hat{\mathcal{X}}^B$ and $\hat{\mathcal{X}}$ satisfy $N_{\hat{\mathcal{X}}^B} + N_{\hat{\mathcal{X}}} = N_{\hat{\mathcal{X}}}$. 
The motivation for introducing these box sets is to simplify the computation while maintaining boundedness for the convergence analysis as we will show later in this paper. 
We split the operator $\optT$ into the following operators $\mathbb{A}$ and $\mathcal{B}$:
\begin{equation}
    \optA: 
    \psi \mapsto (D + \optAy)\psi \;\text{and}\;
    \optB: 
    \psi \mapsto 
    (D + \mathcal{B}_y)\psi
\end{equation}
with $D$, $\mathcal{A}_y$, and $\mathcal{B}_y$ defined by 
\begin{equation}
    D = \begin{bmatrix}
    \frac{\rho_\mu}{2}L_n & \frac{1}{2}(\diagA\mathcal{R})^T & \frac{1}{2}B_n & 0 \\
    -\frac{1}{2}\diagA\mathcal{R} & \frac{\rho_z}{2}L_m & 0 & \frac{1}{2}B_m \\
    -\frac{1}{2}B_n^T & 0 & 0 & 0 \\
    0 & -\frac{1}{2}B_m^T & 0 & 0
    \end{bmatrix}, 
\end{equation}
\begin{equation}
\small
    \optAy : \psi \mapsto \begin{bmatrix}
    \mathcal{R}^T\extgjacob(\by) + N_{\mathcal{\hat{X}^B}}(\by) \\
    \boldsymbol c\\
    0 \\
    0
    \end{bmatrix}
    , \;
    \mathcal{B}_y : \psi \mapsto \begin{bmatrix}
    N_{\mathcal{\hat{X}}}(\by) \\
    N_{\rset{mN}{+}}(\blambda)\\
    0 \\
    0
    \end{bmatrix}. 
\normalsize
\end{equation}

Furthermore, we introduce the following design matrix $\Phi$ for distributedly computing the resolvents $J_{\Phi^{-1}\optA}$ and $J_{\Phi^{-1}\optB}$: 
\begin{equation}\label{eq:dm-phi}
\small
    \Phi = \begin{bmatrix}
    \btau_1^{-1} - \frac{\rho_\mu}{2}L_n & -\frac{1}{2}(\diagA\mathcal{R})^T & -\frac{1}{2}B_n & 0 \\
    -\frac{1}{2}\diagA\mathcal{R} & \btau_2^{-1} - \frac{\rho_z}{2}L_m & 0 & -\frac{1}{2}B_m \\
    -\frac{1}{2}B_n^T & 0 & \btau_3^{-1} & 0 \\
    0 & -\frac{1}{2}B_m^T & 0 & \btau_4^{-1}
    \end{bmatrix},
\normalsize
\end{equation}
where $\btau_1 \coloneqq \blkd{\tau_{11}I_n, \ldots, \tau_{1N}I_n}$ with $\tau_{11} \in \rset{}{++}, \ldots, \tau_{1N} \in \rset{}{++}$; similarly for $\btau_2$, $\btau_3$ and $\btau_4$. 
Notice that these step sizes $\btau_{1}, \ldots, \btau_{4}$ should be small enough to guarantee that $\Phi$ is positive definite. 
Conservative upper bounds for these step sizes can be derived using the Gershgorin circle theorem \cite{bell1965gershgorin}. 

\begin{assumption}\label{asp:phi-pd}
The step sizes $\btau_1, \ldots, \btau_4$ are chosen properly such that the design matrix $\Phi$ in \eqref{eq:dm-phi} is positive definite. \Tblue{Specifically, it suffices to choose 
$\tau_{1i}^{-1} > \frac{1}{2}\norm{A_i}_1 + (\frac{1}{2} + \rho_\mu) d_i$, $\tau_{2i}^{-1} > \frac{1}{2}\norm{A_i}_\infty + (\frac{1}{2} + \rho_z) d_i, \forall i \in \mathcal{N}$, and $\tau_{3j}^{-1} > 1$, $\tau_{4j}^{-1} > 1$, $\forall j \in \edgeE_g$. }
\end{assumption}

\Tblue{Here, $d_i$ denotes the degree of node/player $i$. In general, determining the above step sizes requires some global information acquired through coordination among players such as a proper $\rho_{\mu}$}. 
After the incorporation of the design matrix $\Phi$, we now work in the inner product space $\pspace$ which is a real vector space endowed with the inner product $\langle \psi_1, \psi_2 \rangle_\pspace = \psi_1^T \Phi \psi_2$.
For brevity, let $\poptA \coloneqq \Phi^{-1}\optA$ and $\poptB \coloneqq \Phi^{-1}\optB$. 
In the DR splitting scheme, the general operator $Q$ in \eqref{eq:KM-iter} is given by $\rrcomp_* \coloneqq R_{\poptB} \circ R_{\poptA}$ and it suggests the following exact iteration:
\begin{equation}\label{eq:exact-iter}
    \tilde{\psi}^{(k+1)} \coloneqq \mathscr{P}_*(\tilde{\psi}^{(k)}), \text{ with }\mathscr{P}_* = \idty + \gamma^{(k)}(\rrcomp_* - \idty).
\end{equation} 
Given a generic single-valued operator $Q$, we say that $Q$ is restricted nonexpansive w.r.t. a set $S$ if, for all $\psi \in \dom{Q}$ and $\psi^* \in S$, $\inlnorm{Q\psi - Q\psi^*} \leq \inlnorm{\psi - \psi^*}$ \cite{pavel2019distributed}; 
if, in addition, $S = \fix{Q}$, then $Q$ is quasinonexpansive \cite[Def.~4.1(v)]{BauschkeHeinzH2017CAaM}. 
From the main convergence results in \cite[Thm.~2\&3]{huang2021distributed}, if Assumptions~\ref{asp:subgrad} to \ref{asp:phi-pd} hold, even though $\rrcomp_*$ is not nonexpansive in a general sense, it possesses quasinonexpansiveness in the inner-product space $\pspace$, and hence the sequence $(y_i^{(k)})_{k \in \nset{}{}}$ generated by the exact iteration above (see \cite[Algorithm~1]{huang2021distributed} for detailed implementations) will converge to a v-SGNE of the original problem defined in \eqref{eq:pf-optprob}. 

However, unlike the problem setting in \cite{huang2021distributed} where each player has a closed-form objective function, here the objective function is expected-value, and all too often its closed-form expression may be unavailable. 
Consequently, the $\argmin$ operation in the first player loop of \cite[Algorithm~1]{huang2021distributed} can not be carried out exactly. 
In this case, we need a desirable inexact solver such that, although at each iteration step, it can only get an approximate solution, the computed sequence can still eventually converge to a v-SGNE. 
We let $R_{\appoptA}$ denote the (scenario-based) approximate operator to the exact reflected resolvent $R_{\poptA}$, and $\rrcomp$ denote the corresponding composite $R_{\poptB} \circ R_{\appoptA}$. 
Substituting the operator $\rrcomp_*$ with $\rrcomp$ in \cite[Algorithm~1]{huang2021distributed} gives rise to the following approximate iteration:
\begin{equation}\label{eq:approx-iter}
    \tilde{\psi}^{(k+1)} \coloneqq \mathscr{P}(\tilde{\psi}^{(k)}), \text{ with }\mathscr{P} = \idty + \gamma^{(k)}(\rrcomp - \idty).
\end{equation}
The updating steps of \eqref{eq:approx-iter} are presented in Algorithm~\ref{alg:node-edge}. 
For brevity, let $\supsub{\tilde{y}}{-i(k)}{iL} \coloneqq \sumssub{j \in \neighbN{}{i}}(\supsub{\tilde{y}}{-i(k)}{i} - \supsub{\tilde{y}}{-i(k)}{j})$, and similarly for $\supsub{\tilde{y}}{i(k)}{iL}$, $\supsub{\tilde{\lambda}}{(k)}{iL}$, $\supsub{\hat{y}}{i(k+1)}{iL}$, and $\supsub{\hat{\lambda}}{(k+1)}{iL}$; 
let $\supsub{\tilde{\mu}}{-i(k)}{iB} \coloneqq \sumssub{j \in \neighbN{+}{i}}\supsub{\tilde{\mu}}{-i(k)}{ji} - \sumssub{j \in\neighbN{-}{i}}\supsub{\tilde{\mu}}{-i(k)}{ij}$, and similarly for $\supsub{\tilde{\mu}}{i(k)}{iB}$, $\supsub{\tilde{z}}{(k)}{iB}$, $\supsub{\hat{\mu}}{(k+1)}{iB}$, and $\supsub{\hat{z}}{(k+1)}{iB}$;
let $\supsub{\hat{y}}{(k+1)}{ji} \coloneqq \supsub{\hat{y}}{(k+1)}{i} - \supsub{\hat{y}}{(k+1)}{j}$, and similarly for $\supsub{\hat{\lambda}}{(k+1)}{ji}$, $\supsub{\bar{y}}{(k+1)}{ji}$, and $\supsub{\bar{\lambda}}{(k+1)}{ji}$. 

\begin{algorithm}
\SetAlgoLined
\caption{Distributed v-SGNE Seeking under the Partial-Decision Information Setting}
\label{alg:node-edge}
\textbf{Initialize:} $\{\supsub{\Tilde{y}}{(0)}{i}\}, \{\supsub{\tilde{\lambda}}{(0)}{i}\}, \{\supsub{\tilde{\mu}}{(0)}{ji}\}, \{\supsub{\tilde{z}}{(0)}{ji}\}$\;
\textbf{Iterate until convergence:}\\
\For{player $i \in \playerN$}{

\Tblue{Communicate with neighboring players and edges to obtain $\supsub{\tilde{y}}{(k)}{iL}$, $\supsub{\tilde{\mu}}{(k)}{iB}$, $\supsub{\tilde{\lambda}}{(k)}{iL}$, $\supsub{\tilde{z}}{(k)}{iB}$\;}

$\supsub{y}{-i(k+1)}{i} = \supsub{\Tilde{y}}{-i(k)}{i} - \frac{\supsub{\tau}{}{1i}}{2}(\supsub{\rho}{}{\mu} \supsub{\tilde{y}}{-i(k)}{iL} + \supsub{\tilde{\mu}}{-i(k)}{iB})$ \;

\Tblue{Obtain $y^{i(k+1)}_{i}$ via Subroutine~\ref{alg:proj-stoch-subgrad} that approximately solves: ${\argmin}_{v_i \in \mathcal{X}^B_i}[\supsub{\mathbb{J}}{}{i}(\supsub{v}{}{i}; \supsub{y}{-i(k+1)}{i}) + \frac{1}{2}\supsub{(\Tilde{\lambda}}{(k)}{i})^TA_i\supsub{v}{}{i} $\\
$ \qquad\qquad + \frac{1}{2}{(\supsub{\rho}{}{\mu}\supsub{\tilde{y}}{i(k)}{iL} + \supsub{\mu}{i(k)}{iB})}^T \supsub{v}{}{i} + \frac{1}{2\supsub{\tau}{}{1i}}\norm{\supsub{v}{}{i} - \supsub{\Tilde{y}}{i(k)}{i}}^2]$\;}
$\supsub{\lambda}{(k+1)}{i} = \supsub{\Tilde{\lambda}}{(k)}{i} + \supsub{\tau}{}{2i}(\supsub{A}{}{i}(\supsub{y}{i(k+1)}{i} - \frac{1}{2}\supsub{\Tilde{y}}{i(k)}{i})  -\frac{\supsub{\rho}{}{z}}{2}\supsub{\tilde{\lambda}}{(k)}{iL} -\frac{1}{2}\supsub{\tilde{z}}{(k)}{iB} - \supsub{c}{}{i})$\;
$\supsub{\hat{y}}{(k+1)}{i}=2\supsub{y}{(k+1)}{i} - \supsub{\tilde{y}}{(k)}{i}, \;
\supsub{\hat{\lambda}}{(k+1)}{i}=2\supsub{\lambda}{(k+1)}{i}-\supsub{\tilde{\lambda}}{(k)}{i}$\;
}
\For{edge $(j,i) \in \edgeE_g$}{
\Tblue{Comm. with player $i$ \& $j$ to obtain $\supsub{\hat{y}}{(k+1)}{ji}$, $\supsub{\hat{\lambda}}{(k+1)}{ji}$\;}
$\supsub{\mu}{(k+1)}{ji} = \supsub{\Tilde{\mu}}{(k)}{ji} + \frac{\supsub{\tau}{}{3i}}{2} \supsub{\hat{y}}{(k+1)}{ji},\; 
\supsub{\hat{\mu}}{(k+1)}{ji} = 2\supsub{\mu}{(k+1)}{ji} - \supsub{\tilde{\mu}}{(k)}{ji}$\;
$\supsub{z}{(k+1)}{ji} = \supsub{\Tilde{z}}{(k)}{ji} + \frac{\supsub{\tau}{}{4i}}{2}\supsub{\hat{\lambda}}{(k+1)}{ji}, \; 
\supsub{\hat{z}}{(k+1)}{ji} = 2\supsub{z}{(k+1)}{ji} - \supsub{\tilde{z}}{(k)}{ji}$\;
}

\For{player $i \in \mathcal{N}$}{
\Tblue{Comm. \& obtain $\supsub{\hat{y}}{(k+1)}{iL}$, $\supsub{\hat{\mu}}{(k+1)}{iB}$, $\supsub{\hat{\lambda}}{(k+1)}{iL}$, $\supsub{\hat{z}}{(k+1)}{iB}$\;}
$\supsub{\bar{y}}{(k+1)}{i} = \proj_{\scriptscriptstyle\hat{\mathcal{X}}_i} [\supsub{\hat{y}}{(k+1)}{i}-\frac{\supsub{\tau}{}{1i}}{2}(\supsub{\mathcal{R}}{T}{i}\supsub{A}{T}{i} \supsub{\hat{\lambda}}{(k+1)}{i} + \supsub{\rho}{}{\mu} \supsub{\hat{y}}{(k+1)}{iL} + \supsub{\hat{\mu}}{(k+1)}{iB} )]$\;

$\supsub{\bar{\lambda}}{(k+1)}{i} = \proj_{\scriptscriptstyle\rset{\scriptscriptstyle m}{\scriptscriptstyle +}} [\supsub{\hat{\lambda}}{(k+1)}{i} + \supsub{\tau}{}{2i}
(\supsub{A}{}{i}(\supsub{\bar{y}}{i(k+1)}{i} - \frac{1}{2}\supsub{\hat{y}}{i(k)}{i}) $\\
$ \quad - \frac{\supsub{\rho}{}{z}}{2}\supsub{\hat{\lambda}}{(k+1)}{iL} -\frac{1}{2}\supsub{\hat{z}}{(k+1)}{iB}
)]$\;
}

\For{edge $(j,i) \in \edgeE_g$}{
\Tblue{Comm. \& obtain $\supsub{\bar{y}}{(k+1)}{ji}$, $\supsub{\bar{\lambda}}{(k+1)}{ji}$\;}
$\supsub{\bar{\mu}}{(k+1)}{ji} = \supsub{\hat{\mu}}{(k+1)}{ji} + \supsub{\tau}{}{3i}(\supsub{\bar{y}}{(k+1)}{ji} - \frac{1}{2}\supsub{\hat{y}}{(k+1)}{ji})$\;
$\supsub{\bar{z}}{(k+1)}{ji} = \supsub{\hat{z}}{(k+1)}{ji} + \supsub{\tau}{}{4i}(\supsub{\bar{\lambda}}{(k+1)}{ji} - \frac{1}{2}\supsub{\hat{\lambda}}{(k+1)}{ji})$\;
}
\textbf{K-M updates:} $\supsub{\tilde{\psi}}{(k+1)}{} = \supsub{\tilde{\psi}}{(k)}{} + 2\supsub{\gamma}{(k)}{}(\supsub{\bar{\psi}}{(k+1)}{} - \supsub{\psi}{(k+1)}{});$ \\
\textbf{Return:} \Tblue{$\{\supsub{\tilde{y}}{(k)}i\}$}.
\end{algorithm}
Depending on the inexact solver adopted, $R_{\appoptA}$ usually admits no explicit formulas. 
Yet, as will be shown later in the next section, we can still establish the convergence of Algorithm~\ref{alg:node-edge} based on some specific properties of $R_{\appoptA}$.

\section{Convergence Analysis and Construction of Inexact Solver}\label{sect:convg}

\subsection{General Convergence Results Using Approximate Solution}

We start by stating the Robbins-Siegmund theorem \cite[Thm.~1]{robbins1971convergence}, which plays a significant role in analyzing the convergence of algorithms in the field of stochastic optimization. 
\Tblue{In this subsection, we study the sufficient conditions from a generic perspective to guarantee the convergence of Algorithm~\ref{alg:node-edge} to a v-SGNE of the problem \eqref{eq:pf-optprob}. }
We first define the approximate error and its norm for each iteration as 
\begin{align}
\supsub{\epsilon}{(k)}{} \coloneqq \mathscr{R}(\supsub{\tilde{\psi}}{(k)}{}) - \mathscr{R}_*(\supsub{\tilde{\psi}}{(k)}{}) \;\text{and}\;
\supsub{\varepsilon}{(k)}{} \coloneqq \norm{\supsub{\epsilon}{(k)}{}}_\pspace, 
\end{align}
where $\supsub{\tilde{\psi}}{(k)}{} \coloneqq [\supsub{\tilde{\by}}{(k)}{}; \supsub{\tilde{\blambda}}{(k)}{}; \supsub{\tilde{\bmu}}{(k)}{}; \supsub{\tilde{\bz}}{(k)}{}]$. 
We next introduce the residual function $\res{\tilde{\psi}} \coloneqq \norm{\tilde{\psi} - \rrcomp_*(\tilde{\psi})}_\pspace$ such that $\res{\tilde{\psi}^*} = 0$ is a necessary condition for $\tilde{\psi}^*$ to belong to the fixed-point set of $\rrcomp_*$. 
This relation can be easily checked by using \cite[Prop.~26.1(iii)]{BauschkeHeinzH2017CAaM}.
Let $\mathcal{F}_{k}$ denote the $\sigma$-field comprised of $\{\supsub{\tilde{\psi}}{(0)}{}, \{\supsub{\xi}{(0)}{i}\}_{i \in \playerN}, \ldots, \{\supsub{\xi}{(k-1)}{i}\}_{i \in \playerN}\}$, where for each major iteration $k \in \nset{}{}$, $\supsub{\xi}{(k)}{i} = \{\supsub{\xi}{(k)}{i, 0}, \ldots, \supsub{\xi}{(k)}{i, \supsub{T}{(k)}{i} - 1}\}$ and $\supsub{T}{(k)}{i}$ denotes the number of noise realizations that player $i$ observes at the $k$-th iteration.
\begin{theorem}\label{thm:main-convg-thm}
Consider the SGNEP given in \eqref{eq:pf-optprob}, and suppose Assumptions~\ref{asp:subgrad} to \ref{asp:phi-pd} hold. 
Moreover, $(\gamma^{(k)})_{k \in \nset{}{}}$ is a sequence such that $\gamma^{(k)}\in [0, 1]$ and $\sum_{k \in \nset{}{}}\gamma^{(k)}(1 - \gamma^{(k)}) = +\infty$. 
If the sequence $(\tilde{\psi}^{(k)})$ generated by the inexact solver satisfies
\begin{outline}[enumerate]
    \1 $(\norm{\tilde{\psi}^{(k)}}_\pspace)_{k \in \nset{}{}}$ is bounded a.s.;
    \1 $\sum_{k \in \nset{}{}}\gamma^{(k)}\expt{}{\varepsilon^{(k)} \mid \mathcal{F}^{(k)}} < \infty, \text{ a.s. }$, 
\end{outline}
\Tblue{then $(\tilde{\psi}^{(k)})_{k \in \nset{}{}}$ converges to a fixed point of $\mathscr{R}_*$ a.s., and $\lim_{k \to \infty} J_{\poptA}(\tilde{\psi}^{(k)}) = \psi^*$ a.s. 
Also, the corresponding entries of $\psi^*$ satisfy $\by^* = (\bone_N \otimes y^*)$ and $\blambda^{(k)} = (\bone \otimes \lambda^*)$. 
Here, $y^*$ is a v-SGNE of the original SGNEP \eqref{eq:pf-optprob} and $(y^*, \lambda^*)$ together is a solution to the KKT conditions \eqref{eq:v-sgne} of the SGNEP. }
\end{theorem}
\begin{proof}
See Appendix~\ref{pf:main-convg-thm}. 
\end{proof}

Before proceeding, it is worth highlighting why we need to keep both the condition (i) and (ii) to hold in Theorem~\ref{thm:main-convg-thm}. 
Although the condition (i), i.e., $(\norm{\tilde{\psi}^{(k)}}_\pspace)_{k \in \nset{}{}}$ is bounded a.s., is a necessary condition for the summability statement in (ii), as has been showed in \cite[Prop.~5.34]{BauschkeHeinzH2017CAaM} for deterministic cases, under the partial-information setting, a natural strategy is to prove the condition (i) first using a more primitive condition, and then establish the condition (ii) based on (i). 
\Tblue{The specific conditions regarding the algorithm parameters to ensure the convergence will be later discussed in Theorem~\ref{thm:summb}.} 


\begin{remark}
When proving Theorem~\ref{thm:main-convg-thm}, the inequalities invoked follow from the quasinonexpansiveness of the exact operator $\rrcomp_*$ and the Cauchy-Schwarz inequality.
The proof and conclusion in Theorem~\ref{thm:main-convg-thm} thus can be applied to the analysis of a general \Tblue{continuous} operator $Q$ in \eqref{eq:KM-iter} and its approximation other than the operators $\rrcomp_*$ and $\rrcomp$ in this paper, as long as the operator $Q$ is quasinonexpansive and the conditions regarding $(\gamma^{(k)})_{k \in \nset{}{}}$, $(\varepsilon^{(k)})_{k \in \nset{}{}}$, and $(\tilde{\psi}^{(k)})_{k \in \nset{}{}}$ in Theorem~\ref{thm:main-convg-thm} are satisfied. 
\end{remark}

\subsection{Construction of a Desirable Inexact Solver}

As we discussed at the end of Section~\ref{sect:dr-algm}, it is challenging to solve the augmented best-response subproblems that involve the exact expected-value objectives (the $\argmin$ problems in the first player for-loop of Algorithm~\ref{alg:node-edge}). 
Theorem~\ref{thm:main-convg-thm} suggests that we can still obtain a v-SGNE by solving these augmented best-response subproblems not precisely but up to some prescribed accuracy. 
In this subsection, we consider a specific scenario-based solver using the projected stochastic subgradient method \cite[Ch.~2]{shor2012minimization}.
As has been shown in the existing literature\cite{lacoste2012simpler}, the weighted average of the projected stochastic subgradient method possesses an $O(1/t)$ convergence rate if the subgradient is unbiased and the variance of the subgradient is finite. 
Here, we study the explicit conditions that the projected stochastic subgradient solver should satisfy to serve as a feasible inexact solver in the context of distributed SGNEP with only partial-decision information, as suggested in Theorem~\ref{thm:main-convg-thm}. 

We first assume the unbiasedness and finite-variance properties of a general projected stochastic subgradient method. 
Throughout this subsection, we use $k$ to index the major iterations (the iteration of the v-SGNE seeking Algorithm~\ref{alg:node-edge}) and $t$ to index the minor iterations (the iteration of the inexact solver in the first player for-loop of Algorithm~\ref{alg:node-edge}). 
Furthermore, at each major iteration $k$, for each player $i$, 
let the augmented scenario-based objective function be denoted by 
$\hat{J}^{(k)}_{i}(v_i; \xi^{(k)}_{i,t}) \coloneqq J_i(v_i; y^{-i(k+1)}_i, \xi^{(k)}_{i,t}) + (\tilde{\varphi}^{(k)}_i)^Ty^i_i + \frac{1}{2\tau_{1i}}\norm{v_i - \tilde{y}^{i(k)}_i}^2_2$, 
and the augmented expected-value objective function be denoted by 
$\hat{\mathbb{J}}^{(k)}_i(v_i) \coloneqq \mathbb{J}_i(v_i; y^{-i(k+1)}_i) + (\tilde{\varphi}^{(k)}_i)^Tv_i + \frac{1}{2\tau_{1i}}\norm{v_i - \tilde{y}^{i(k)}_i}^2_2$, 
where $\tilde{\varphi}^{(k)}_i \coloneqq \frac{1}{2}(A_i^T\tilde{\lambda}^{(k)}_i + \tilde{\mu}^{i(k)}_{iB} + \rho_{\mu}\tilde{y}^{i(k)}_{iL})$. 
Note that $\hat{\mathbb{J}}^{(k)}_i(\cdot)$ is the objective in the first player-loop of Algorithm~\ref{alg:node-edge} that needs to be inexactly solved. 
Here, the vector $\tilde{\varphi}^{(k)}_i$ represents some augmented terms that enforce the consensus constraints and the global resource constraints. 
For brevity, the local estimates of the other players' decisions $y^{-i(k+1)}_i$ are omitted from the arguments of the augmented functions defined above. 
Let $T^{(k)}_i$ denote the total number of the projected stochastic subgradient steps taken in the $k$-th major iteration by player $i$.
The subgradient of the scenario-based objective function at the $k$-th major iteration and the $t$-th minor iteration is denoted by $g^{(k)}_{i,t} \in \partial_{y^i_i}\hat{J}^{(k)}_i(y^{i(k+1)}_{i,t}; \xi^{(k)}_{i,t})$, where $t = 0, 1, \ldots, T^{(k)}_i-1$. 


\begin{assumption}\label{asp:proj-stoch-subgrad}
\Tblue{For each player $i \in \playerN$, at each major iteration $k$ of Algorithm~\ref{alg:node-edge}, the following statements hold:\\
(i) (Unbiasedness) At each minor iteration $t$, there exists a $g^{(k)}_{i,t} \in \partial_{y^i_i}\hat{J}^{(k)}_i(y^{i(k+1)}_{i,t}; \xi^{(k)}_{i,t})$ such that $\expt{}{g^{(k)}_{i,t} \mid \sigma\{\mathcal{F}_k, \xi^{(k)}_{i, [t]}\}}$ is a.s. a subgradient of the expected-value augmented objective $\hat{\mathbb{J}}_i^{(k)}(\cdot)$ at $y^{i(k+1)}_{i,t}$, where $\xi^{(k)}_{i, [t]} \coloneqq \{ \xi^{(k)}_{i, 0}, \ldots, \xi^{(k)}_{i,t-1} \}$ with $\xi^{(k)}_{i, [0]} \coloneqq \varnothing$; \\
(ii) (Upper-bounded variance) For any $y^i_i \in \mathcal{X}^B_i$, there exists a $g^{(k)}_i \in \partial_{y^i_i}\hat{J}^{(k)}_i(y^i_i; \xi_i)$ such that  $\expt{}{\norm{g^{(k)}_{i}}^2_2 \mid \mathcal{F}_k} \leq \alpha_{g,i}^2\norm{\tilde{\psi}^{(k)}}^2_2 + \beta_{g,i}^2$ a.s. for some positive constants $\alpha_{g,i}$ and $\beta_{g,i}$.}
\end{assumption}

We refer the reader to the paragraph before Theorem~\ref{thm:main-convg-thm} for the definitions of the stack vector $\tilde{\psi}^{(k)}$ and the filtration $(\mathcal{F}_k)_{k \in \nset{}{}}$ as a reminder.
The proposed inexact solver for the first player for-loop of Algorithm~\ref{alg:node-edge} is given in Subroutine~\ref{alg:proj-stoch-subgrad}. 
\Tblue{Note that the DR splitting scheme ensures local feasibility with single projection onto local feasible sets, and requires multiple projections onto relaxed bounded box sets, which considerable reduces the computational complexity compared with other methods such as proximal-point scheme \cite{bianchi2022fast}.}

\begin{algorithm}
\SetAlgorithmName{Subroutine}{subroutine}{List of Subroutines}
\SetAlgoLined
\caption{\Tblue{Projected Stochastic Subgradient Solver}}
\label{alg:proj-stoch-subgrad}
\Tblue{\textbf{For each $i \in \playerN$, at the $k$-th major iteration of Alg.~\ref{alg:node-edge}:}} \\
\textbf{Initialize:} $y^{i(k+1)}_{i,0} \coloneqq \tilde{y}^{i(k)}_{i}; $\\
\For{$t = 0 \text{ to } T^{(k)}_i - 1$}{
$y^{i(k+1)}_{i,t+1} \coloneqq \proj_{\mathcal{X}^B_i}[y^{i(k+1)}_{i,t} - \kappa_{i,t} \cdot g^{(k)}_{i,t}]$, 
with $\kappa_{i,t} \coloneqq \frac{2\tau_{1i}}{t+2}$;\\
}
\textbf{Return:} $y^{i(k+1)}_{i} \coloneqq y^{i(k+1)}_{i, T^{(k)}_{i}}$.
\end{algorithm}

The following lemma discusses the convergence rate of Subroutine~\ref{alg:proj-stoch-subgrad} as a minor updating routine inside Algorithm~\ref{alg:node-edge}. 
We use $y^{i(k+1)}_{i, *}$ to denote the accurate minimizer of the expected-value augmented function $\hat{\mathbb{J}}^{(k)}_i(\cdot)$.

\begin{lemma}\label{le:proj-stoch-convg-rate}
Suppose Assumptions~\ref{asp:subgrad} to \ref{asp:proj-stoch-subgrad} hold. 
Then, for any $T = 1, \ldots, T^{(k)}_{i}$, the distance between the approximate solution by Subroutine~\ref{alg:proj-stoch-subgrad} and the accurate solution satisfies $\expt{}{\norm{y^{i(k+1)}_{i,T} - y^{i(k+1)}_{i,*}}^2_2 \mid \mathcal{F}_k} \leq 4\tau_{1i}^2T^{-1}(\alpha^2_{g,i}\norm{\tilde{\psi}^{(k)}}^2_2 + \beta^2_{g,i})$ a.s. 
\end{lemma}
\begin{proof}
See Appendix~\ref{pf:proj-stoch-convg-rate}.
\end{proof}

From Lemma~\ref{le:proj-stoch-convg-rate}, we can conclude that for each player $i \in \playerN$, after the $k$-th major iteration of Algorithm~\ref{alg:node-edge} where player $i$ implements $T^{(k)}_i$ projected stochastic subgradient steps in Subroutine~\ref{alg:proj-stoch-subgrad}, $\bexpt{}{\norm{y^{i(k+1)}_{i} - y^{i(k+1)}_{i,*}}^2_2 \mid \mathcal{F}_k} \leq \frac{(2\tau_{1i})^2}{T^{(k)}_i}(\alpha^2_{g,i}\norm{\tilde{\psi}^{(k)}}^2_2 + \beta^2_{g,i})$. 
Based on this result, it is straightforward to derive an upper bound for the approximate error $\varepsilon^{(k)} \coloneqq \norm{\rrcomp(\tilde{\psi}^{(k)}) - \rrcomp_*(\tilde{\psi}^{(k)})}_\pspace$. 
As will be shown later, this upper bound can be treated as a function of $\ubar{T}^{(k)} \coloneqq \min\{T^{(k)}_{i}: i \in \playerN\}$ which we can tune to provide a desirable sequence of approximation accuracies.

\begin{lemma}\label{le:convg-rate-augvec}
Consider $(\varepsilon^{(k)})_{k \in \nset{}{}}$ generated by Algorithm~\ref{alg:node-edge} using Subroutine~\ref{alg:proj-stoch-subgrad} as the inexact solver. 
Suppose Assumptions~\ref{asp:subgrad} to \ref{asp:proj-stoch-subgrad} hold. 
Then there exist some positive constants $\alpha_\psi$ and $\beta_\psi$ such that the following relation holds a.s.:
\begin{equation}
\bexpt{}{\varepsilon^{(k)} \mid \mathcal{F}_k} \leq (\ubar{T}^{(k)})^{-1/2}(\alpha_{\psi}\norm{\tilde{\psi}^{(k)}}_\pspace + \beta_\psi).
\end{equation}
\end{lemma}
\begin{proof}
See Appendix~\ref{pf:convg-rate-augvec}.
\end{proof}

Lemma~\ref{le:convg-rate-augvec} establishes the relationship between the approximate error $\varepsilon^{(k)}$ and the stack vector $\tilde{\psi}^{(k)}$ at each major iteration $k$. 
We define $\gamma^{(k)}_T \coloneqq \gamma^{(k)}(\ubar{T}^{(k)})^{-1/2}$. 
From Theorem~\ref{thm:main-convg-thm}, it suffices to have the sequence $(\gamma^{(k)}_T)_{k \in \nset{}{}}$ summable and $(\norm{\tilde{\psi}^{(k)}})_{k \in \nset{}{}}$ bounded. 
To this end, we next focus on proving the conditions needed to guarantee the boundedness of $(\tilde{\psi}^{(k)})_{k \in \nset{}{}}$ and finally derive the sufficient conditions to ensure the convergence of Algorithm~\ref{alg:node-edge}. 

\begin{theorem}\label{thm:summb}
Consider the sequence $(\tilde{\psi}^{(k)})_{k \in \nset{}{}}$ generated by Algorithm~\ref{alg:node-edge} using Subroutine~\ref{alg:proj-stoch-subgrad} as an inexact solver. 
Suppose Assumptions \ref{asp:subgrad} to \ref{asp:proj-stoch-subgrad} hold. 
In addition, the sequence $(\supsub{\gamma}{(k)}{})_{k \in \nset{}{}}$ satisfies $0 \leq \gamma^{(k)} \leq 1$ and $\sum_{k \in \nset{}{}}\gamma^{(k)}(1 - \gamma^{(k)}) = +\infty$, and the sequence $(\gamma^{(k)}_T)_{k \in \nset{}{}}$ is absolutely summable. 
Then $(\norm{\tilde{\psi}^{(k)}}_\pspace)_{k \in \nset{}{}}$ is bounded a.s., and 
$\sum_{k \in \nset{}{}} \gamma^{(k)}\expt{}{\varepsilon^{(k)} \mid \mathcal{F}_k} < \infty$ a.s. 
\Tblue{As a result, the sequence $(\tilde{\psi}^{(k)})_{k \in \nset{}{}}$ will converge to a fixed point of $\mathscr{R}_*$ and the associated sequence $(y^{(k)})_{k \in \nset{}{}}$ generated by $J_{\poptA}(\tilde{\psi}^{(k)})$ will converge to a v-SGNE of the problem \eqref{eq:pf-optprob}. }
\end{theorem}
\begin{proof}
See Appendix~\ref{pf:summb}.
\end{proof}

\Tblue{Based on Theorem~\ref{thm:summb}, to get a solution arbitrarily closed to a v-SGNE, we will run Algorithm~\ref{alg:node-edge} for a sufficiently large number of major iterations. 
Then, we use the obtained last-iterate $\tilde{\psi}^{(k)}$ to run Subroutine~\ref{alg:proj-stoch-subgrad} for another sufficiently large number of minor iterations $\ubar{T}^{(k)}$ \cite[Lemma~3]{huang2021sdistributed}.}
To ensure the convergence of Algorithm~\ref{alg:node-edge}, it suffices to properly choose $(\gamma^{(k)})_{k \in \nset{}{}}$ and $(\ubar{T}^{(k)})_{k \in \nset{}{}}$ such that $(\gamma^{(k)}_{T})_{k \in \nset{}{}}$ is a summable sequence.
\Tblue{As an example for admissible parameters, we can choose $\gamma^{(k)} = 1/k^a$ and $\ubar{T}^{(k)} = k^b$, with $0 < a \leq 1$ and $a+b/2 > 1$.} 
This can be manipulated to make the proposed algorithm work under different practical settings. 
For instance, if these players are working in a feedback-parsimonious setting, i.e., the available realizations of noisy first-order/gradient information per iteration are scarce, one can choose a faster decaying rate for $(\gamma^{(k)})_{k \in \nset{}{}}$ as long as $\sum_{k \in \nset{}{}} \gamma^{(k)} = +\infty$ and let $(\ubar{T}^{(k)})_{k \in \nset{}{}}$ grow linearly or even sublinearly. 
In contrast, if the available realizations are abundant, one can let $(\ubar{T}^{(k)})_{k \in \nset{}{}}$ grow superlinearly while fixing $\gamma^{(k)}$ to be some constant such that the proposed algorithm can enjoy a faster convergence rate.

\section{Case Study and Numerical Simulations}\label{sect:simu}

\subsection{Stochastic Nash-Cournot Distribution Game}\label{subsect:nash-cournot-dist}

We evaluate the performance of the proposed algorithm with a Nash Cournot distribution problem \cite[Sec.~1.4.3]{facchinei2007finite}\cite{parise2019distributed} over a transport network. 
Several firms (indexed by $\playerN \coloneqq \{1, \ldots, N\}$), who produce a common homogeneous commodity, participate in this game.
These firms try to optimize their own payoffs by deciding the quantity of the commodity to produce at each factory and the quantities to distribute to different markets. 
A transport network is provided, with markets as the nodes and roads as the edges. 
Let $\mathcal{N}_T$ denote the node set of this network and $\mathcal{E}_T$ the edge set, distinguished from $\mathcal{N}_g$ and $\mathcal{E}_g$ of the underlying communication network $\mathcal{G}$. 
Denote cardinalities of $\mathcal{N}_T$ and $\mathcal{E}_T$ by $N_T$ and $E_T$, and the incident matrix of this transport network by $B_T \in \rset{N_T \times E_T}{}$. 

Each firm has $N_{T_i}$ factories at certain nodes on this transport network, given by the set $\mathcal{N}_{T_i}$. 
Its decision vector $x_i \in \rset{E_T + N_{T_i}}{}$ is comprised of two parts ($x_i \coloneqq [u_i; v_i]$): 
each entry of $u_i \in \rset{E_T}{+}$ represents the quantity of the commodity delivered through a road in $\mathcal{E}_T$; 
each entry of $v_i \in \rset{N_{T_i}}{+}$ represents the quantity of the commodity produced by one of its factories in $\mathcal{N}_{T_i}$. 
The indicator matrix which maps from each entry of $v_i$ to the corresponding node on the transport network is denoted by $E_i \in \rset{N_T \times N_{T_i}}{}$, and we let $A_i \coloneqq [B_T, E_i]$. 
These two parts ($u_i$ and $v_i$) together uniquely determine the distribution of commodity $A_ix_i$ over the markets. 
If we assume that the factories owned by firm $i$ have maximum production capacities $b_i \in \rset{N_{T_i}}{++}$, then each entry of the vector $u_i \in \rset{E_T}{}$ is upper-bounded by $\norm{b_i}_1$, and the local feasible set $\mathcal{X}_i$ is a polytope which can be written as: 
$\mathcal{X}_i \coloneqq \{x_i \in \rset{E_T + N_{T_i}}{} \mid 0 \leq v_i \leq b_i, 0 \leq u_i \leq \norm{b_i}_1 \otimes \bone_{E_T}, A_ix_i \geq 0\}$.
\Tblue{The objective function of each firm $i$ is given by:
$J_i(x_i ;x_{-i}, \xi_i) = x_i^TQ_ix_i + C_t(u_i) + C^i_p(v_i) - (P(Ax) + \xi_i)^TA_ix_i$, 
where $Q_i \in \pdset{E_T + N_{T_i}}{++}$, $A \coloneqq [A_1, \ldots, A_N]$, $x \coloneqq [x_1; \ldots; x_N] \in \rset{n}{}$ with $n \coloneqq NE_T + \sum_{i \in \playerN}N_{T_i}$, and $P(Ax) \coloneqq w - \Sigma Ax$ maps from the total quantities $Ax$ of the commodity at markets to their unit prices with $w \in \rset{N_T}{+}$ and $\Sigma \in \pdset{N_T}{++}$.} 
The transport cost $C_t$ is defined as the sum of the costs at all roads, i.e., $C_t(u_i) \coloneqq \sum_{k \in \mathcal{E}_T}C^k_t([u_i]_k)$, where each road $k \in \mathcal{E}_T$ has $C^k_t([u_i]_k) \coloneqq \eta_k([u_i]_k - (1 - \frac{1}{1 + [u_i]_k}))$. 
The production cost $C^i_p$ is also defined as the sum of the costs at all factories, i.e., $C^i_p(v_i) \coloneqq \sum_{k \in \mathcal{N}_{T_i}} C^{i,k}_p(v_i)$, where each factory $k \in \mathcal{N}_{T_i}$ has $C^{i,k}_p([v_i]_k) \coloneqq \kappa_{i,k}([v_i]_k-(1 - \frac{1}{1 + [v_i]_k}))$. 
The total income $(P(Ax) + \xi_i)^TA_ix_i$ captures uncertainty in the unit prices through the random vector $\xi_i$, which has its entries independently identically distributed with mean zero. 

Furthermore, we assume that each market has a maximum capacity for the commodity, and the decision vectors of the players should collectively satisfy the global resource constraints $\sum_{i \in \playerN} A_ix_i \leq c$ where $c \in \rset{N_T}{++}$. 
Building on the discussed setups, each firm $i \in \playerN$, given the production and distribution strategies of the other players ($x_{-i}$), aims to solve the following stochastic optimization problem:
\begin{equation}
\begin{cases}
\minimize_{x_i \in \mathcal{X}_i} & \expt{\xi_i}{J_i(x_i; x_{-i}, \xi_i)} \\
\subj & A_ix_i \leq c - \sum_{j \in \neighbN{}{-i}}A_jx_j. 
\end{cases}
\end{equation}

\subsubsection{Assumptions Verification}

We use the transport network of the city of Oldenburg \cite{brinkhoff2002framework} (Fig.~\ref{fig:nash-dist-ne} top): it consists of $N_T = 29$ nodes (markets) and $E_T = 2 \times 34$ directed edges (roads). 
Five firms ($N = 5$) participates in this game, and each firm has a single factory at a given location/node $\{8, 14, 21, 10, 29\}$. 
\Tblue{Each factory has its maximum production capacity uniformly sampled from the interval $[10, 14]$, and $Q_i$ is a diagonal matrix with the diagonal entries uniformly sampled from $[2, 3]$}. 
In the transport costs, we have $\frac{1}{8}\eta_{k} \in (0, 1]$ being the ratio between the length of road $k$ and the maximum length of the roads in $\mathcal{E}_T$. 
In the production costs, we fix the coefficients $\kappa_{i,k} = 2$. 
In the price function $P(\cdot)$, we draw each entry of the vector $w$ uniformly at random from the interval $[26, 30]$ and set the matrix $\Sigma$ to have $[\Sigma]_{ii} \coloneqq 1$ for all $i \in \mathcal{N}_T$ and $[\Sigma]_{ji} \coloneqq 0.3 \cdot (1 - \frac{1}{8}\eta_{(j,i)})$ for all $(j, i) \in \mathcal{E}_T$. 

For each player $i \in \playerN$, it is easy to check by definition that $J_i(x_i; x_{-i}, \xi_i)$ and $\mathbb{J}_i(x_i; x_{-i})$ are smooth and proper, and they are convex in $x_i$. 
Moreover, the pseudogradient $\mathbb{F}$ is strongly monotone on the local compact feasible sets $\prod_{i \in \playerN} \mathcal{X}_i$ (detailed verifications are omitted due to space limit). 
Then by \cite[Thm.~2.3.3]{facchinei2007finite}, this problem admits a unique v-SGNE. 
We set the communication graph of the players to be composed of an undirected circle plus two randomly selected edges.
Therefore, Assumptions~\ref{asp:subgrad} to \ref{asp:commtopo} are fulfilled. 
\Tblue{We choose $\rho_\mu = 8$ and then appropriately set the step sizes to be $\btau_1 = 0.0285 \otimes I_{Nn}$, $\btau_2 = 0.09 \otimes I_{Nm}$, $\btau_3 = 0.5 \otimes I_{En}$, and $\btau_4 = 0.5 \otimes I_{Em}$. 
It can be checked numerically that the conditions in Assumptions~\ref{asp:convg} and \ref{asp:phi-pd} are satisfied. 
We further set $\xi_i \sim U[-2, 2]$ and can easily verify the conditions in Assumption~\ref{asp:proj-stoch-subgrad}.} 

\subsubsection{Simulation Results}\label{subsubsect:nash-dist-result}

\Tblue{The sequence $(\gamma^{(k)})_{k \in \nset{}{}}$ is fixed to be $1/2$, and the subgradient steps taken is chosen as $T^{(k)} = \ceil{10^{-4}k^{2.1}} + 20$. 
We compare the performance of Algorithm~\ref{alg:node-edge} with that of \cite{franci2020stochastic}, with $c=4$ and the relaxed step sizes chosen as $0.04$.
Theses step sizes are empirically pushed to near the upper limit of convergence; otherwise, \cite[Lemma~6]{pavel2019distributed} suggests a set of miniscule and conservative step sizes ($\approx 3\times 10^{-5}$).} 
The performances of the proposed algorithm are shown in Fig.~\ref{fig:nash-cournot-dist}. 
We use the thick and semi-transparent lines to illustrate the real fluctuations of the metrics throughout the iterations, while using the thin lines to exhibit the simple moving averages of the metrics with a window size of $30$. 
The averages of the normalized distances to the v-SGNE are presented in Fig.~\ref{fig:nash-cournot-dist}(a), where the unique v-SGNE is calculated using the centralized method from \cite{belgioioso2018projected}.  
Note that $y^{(k)}_j$ denotes the stack of player $j$'s local decision and local estimates at the $k$-th iteration, and $y^*$ the v-SGNE of the game.
Fig.~\ref{fig:nash-cournot-dist}(b) shows the relative lengths of the updating step at each iteration. 
Let $\bar{y}^{(k)} \coloneqq \frac{1}{N}\sum_{j \in \playerN}y^{(k)}_j$. 
Fig.~\ref{fig:nash-cournot-dist}(c) exhibits how the sums of the standard deviations of the local estimates $\{y_j\}$, i.e., $\sum_{\ell=1}^{n}(\frac{1}{N}\sum_{j \in \playerN}([y^{(k)}_j]_\ell - [\Bar{y}^{(k)}]_\ell)^2)^{\frac{1}{2}}$, evolve over the iterations. It measures the level of consensus among different local estimates $y_j$. Fig.~\ref{fig:nash-cournot-dist}(d) is almost the same as Fig.~\ref{fig:nash-cournot-dist}(c) except that we are now investigating the consensus of local dual variables $\{\lambda_j\}$. 
The computed v-SGNE of this problem is illustrated in Fig.~\ref{fig:nash-dist-ne}, where we use five different colors to represent the different players/firms. 
The top panel includes a geographic illustration, with the locations of the factories denoted by the colored letters and the total quantities transported on the roads illustrated by the brightness of the edges. 
The bottom panel shows the commodity contributions from the players at each market on this transport network. 

\begin{figure}
    \centering
    \includegraphics[width=0.42\textwidth]{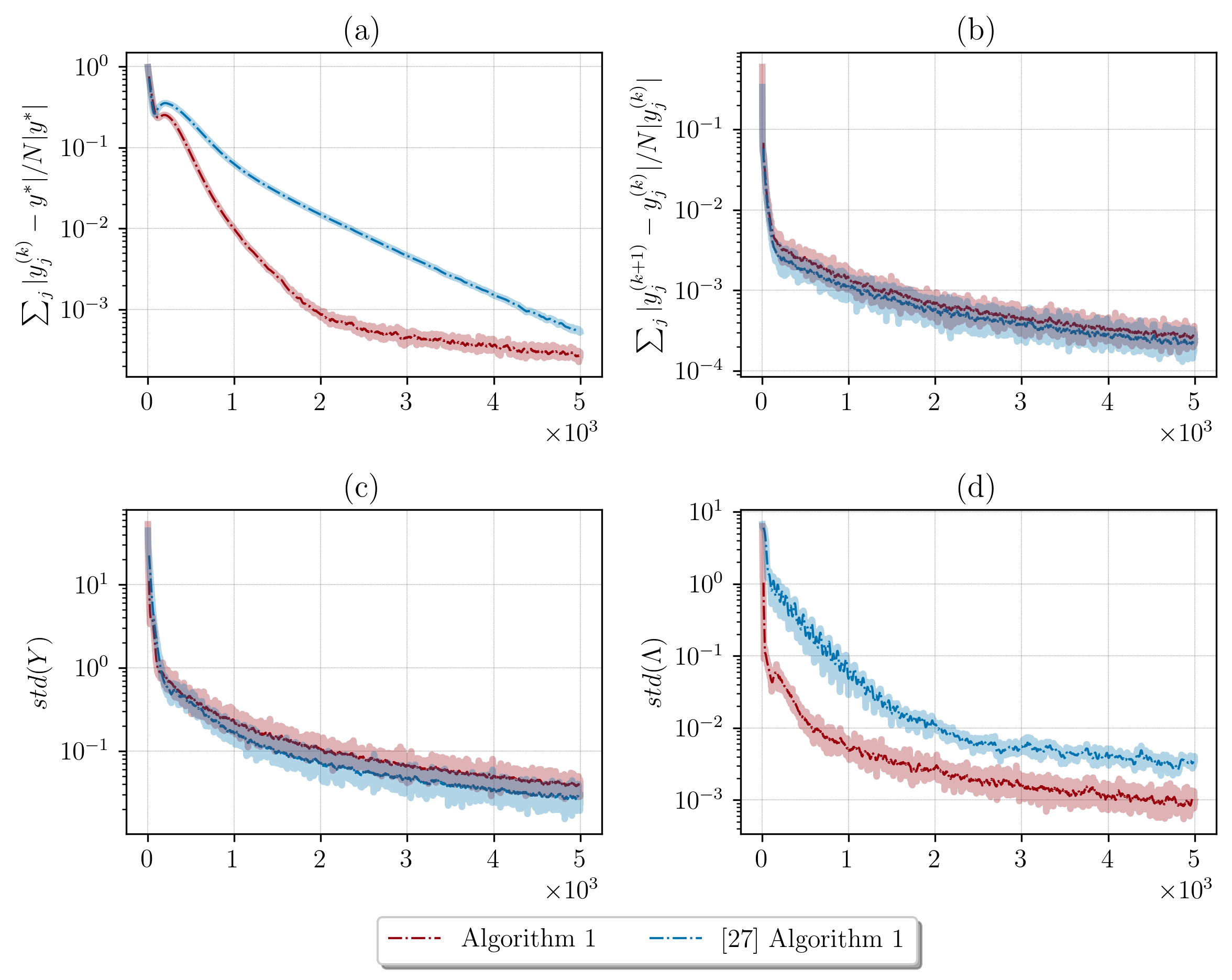}
    \caption{Performances of Alg.~\ref{alg:node-edge} in a Nash-Cournot Game}
    \label{fig:nash-cournot-dist}
\end{figure}

\begin{figure}
    \centering
    \includegraphics[width=0.4\textwidth]{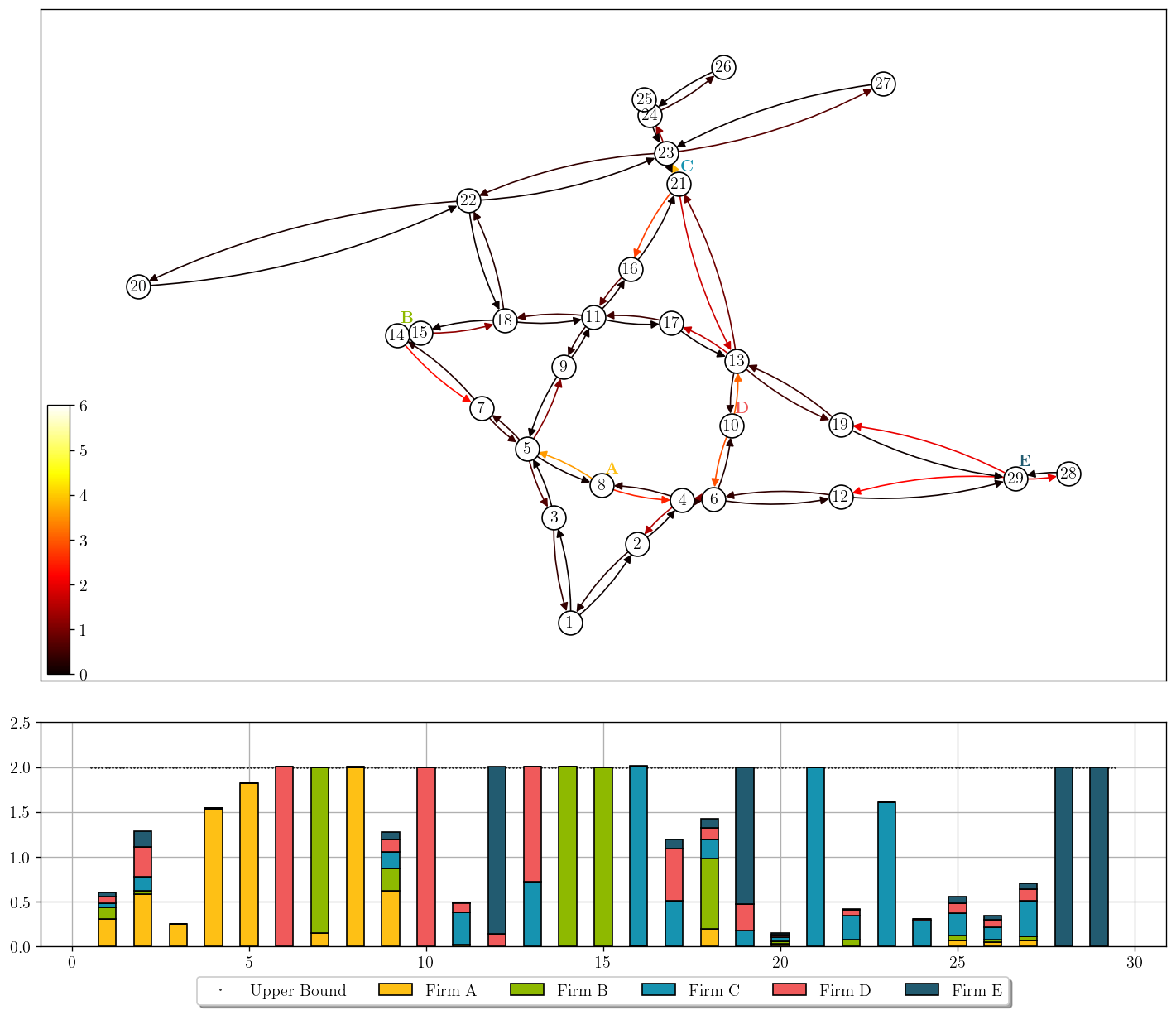}
    \caption{The v-SGNE Obtained by Alg.~\ref{alg:node-edge}}
    \label{fig:nash-dist-ne}
\end{figure}

\subsection{Multi-Product Assembly Game with the Two-Stage Model}

The two-stage stochastic programming problem originated from the work of \cite{dantzig1955linear} and found its applications in fields such as financial planning and control \cite[Sec.~1.2]{birge2011introduction}, investment in power plants \cite[Sec.~1.3]{birge2011introduction}, transportation planning during emergency response \cite{barbarosoglu2004two}, etc.
In this paper, we consider a multi-product assembly problem using the two-stage model \cite[Sec.~1.3.1]{shapiro2014lectures}. 
In a game network with $N$ manufacturers/players indexed by $\playerN \coloneqq \{1, \ldots, N\}$, each player $i$ produces $\ell_i$ types of commodities. 
There are in total $m$ different subassemblies which have to be ordered from a third-party vendor. 
For each player $i$, it needs $n_i$ different types of subassemblies in total, and a unit of commodity $j$ requires $h_{i, (j,v)}$ units of subassembly $v$, where $j=1, \ldots, \ell_i$ and $v = 1, \ldots, n_i$. 
The demands for player $i$'s commodities are modeled as a random vector $D_{i} \coloneqq [D_{i, 1}; \cdots; D_{i, \ell_i}]$, which has its range $\mathcal{D}_i$ inside a bounded set in the positive orthant. 

We start by formulating the second-stage problem. 
Let the numbers of subassemblies ordered by player $i$ be denoted by $x_{i} \in \rset{n_i}{+}$, which is treated as a parameter in the second-stage problem. 
In this stage, player $i$ makes a production plan about the quantity of each commodity to produce based on the realized demand vector $d_i \in \rset{\ell_i}{+}$. 
This production plan should maximize the profit and at the same time not exceed the quantities of available subassemblies. 
The income of player $i$ is comprised of the unit selling prices of the commodities $p_{i} \in \rset{\ell_i}{}$ and the unit salvage values of subassemblies that are not used $s \in \rset{m}{}$. 
Denote the numbers of produced units by $z_{i} \in \rset{\ell_i}{+}$, and the numbers of subassemblies left in inventory by $y_{i} \in \rset{n_i}{+}$. 
We introduce the matrix $H_i \in \rset{\ell_i \times n_i}{}$ with each entry $[H_i]_{(j, v)} = h_{i, (j,v)}$ and a binary matrix $A_i \in \rset{m \times n_i}{}$ mapping each entry of $y_{i}$ to one among the $m$ subassemblies. 
In addition, assume the full-row-rank matrix $H_i$ has $\ell_i \leq n_i$ and no column sums to zero. 
Then we can define the nonsmooth function $\mathcal{Q}_i(x_i; d_i) = \min\{-p_i^Tz_i-s^TA_iy_i \mid  y_i = x_i - H_i^Tz_i, \bzero \leq z_i \leq d_i, y_i \geq \bzero\}$, the minimizer of which is the best production plan. 

With $\mathcal{Q}_i(x_i; d_i)$ defined, we can then formulate the first-stage problem. 
The price of subassembly $v$ per unit consists of the base cost $C_v$ which is a random variable and the additional cost with the increasing ratio $[\Sigma]_{(\nu,\nu)}$ per ordered unit. 
At this stage, when making decisions about the pre-order quantities $x_i$ to maximize the profit, each player $i$ is uncertain about the base prices of subassemblies and the demands for its commodities. 
Each player $i$ has an expected-value objective w.r.t. the random vectors $C\coloneqq [C_\nu]_{\nu = 1, \ldots, m}$ and $D_i$. 
Moreover, their decisions should collectively satisfy the global constraints concerning the available subassemblies. 
Altogether, the first-stage problem for each player $i$ can be expressed as:
\begin{equation}\label{eq:1st-stage}
\begin{cases}
\underset{x_i \in \mathcal{X}_i}{\minimize} & \expt{}{\frac{1}{2}x_i^TQ_ix_i + (C + \Sigma Ax)^TA_ix_i + \mathcal{Q}_i(x_i; D_i)} \\
\subj & A_ix_i \leq c - \sum_{j \in \neighbN{}{-i}}A_jx_j,
\end{cases}
\end{equation}
where $A \coloneqq [A_1, \ldots, A_N]$, $x \coloneqq [x_1; \ldots; x_N]$, 
$\mathcal{X}_i$ is the local feasible set of the decision vector $x_i$ which is compact and convex, 
$Q_i$ and $\Sigma$ are diagonal matrices with each diagonal entry positive, 
and the constant vector $c \in \rset{m}{}$ denotes the quantities of available subassemblies. 

Suppose $N = 5$ players participate in this game to compete for $m = 10$ types of subassemblies. 
The decision vector of each player $i$ has dimension $n_i$ chosen uniformly at random from $\{7, 8, 9, 10\}$.
The local feasible set $\mathcal{X}_i$ is the direct product of $n_i$ connected compact intervals. 
The communication graph consists of a directed circle and two randomly selected edges. 

\subsubsection{Assumptions Verification}

We claim that the function $\mathcal{Q}_i(x_i; d_i)$ is a piecewise linear function in $x_i \in \mathcal{X}_i$ given any fixed $d_i \in \mathcal{D}_i$, where $\mathcal{D}_i$ and $\mathcal{X}_i$ are both bounded. 
We first introduce the residual variable $r_i = d_i - z_i$ and convert the inequality constraints in $\mathcal{Q}_i(x_i; d_i)$ to equality ones as follows:
\begin{equation}
\begin{cases}
\minimize_{z_i\geq \bzero, h_i\geq \bzero, r_i \geq \bzero} \; -p_i^Tz_i - s^TA_ih_i\\
\subj \; h_i = x_i - H_i^Tz_i, z_i + r_i = d_i. 
\end{cases}
\end{equation}

By letting 
$B_i = \Big[\begin{smallmatrix} H_i^T & I_{n_i} & \bzero_{n_i \times \ell_i} \\ I_{\ell_i} & \bzero_{\ell_i \times n_i} & I_{\ell_i} \end{smallmatrix}\Big]$, 
$u_i \coloneqq [z_i; h_i; r_i]$, 
$q_i \coloneqq [-p_i; -A_i^Ts; \bzero_{\ell_i}]$, 
$\Tilde{I}_i = [I_{n_i}; \bzero_{\ell_i \times n_i}]$, and 
$\Tilde{d}_i = [\bzero_{n_i}; d_i]$, the above constrained linear programming can be presented as:
$\minimize_{u_i} q_i^T u_i$, while $\subj B_iu_i = \tilde{I}_ix_i + \tilde{d}_i$ and $u_i \geq \bzero$. 
Its dual problem can then be derived as:
\begin{equation}\label{eq:2nd-stage-dual}
\maximize_{v_i} (\tilde{I}_ix_i + \tilde{d}_i)^T v_i, \;
\subj B_i^Tv_i \leq q_i.
\end{equation}
We progress with the dual problem which only has $x_i$ as the coefficients of the objective function.  
Since the feasible set $\mathcal{X}_i$ is compact inside the non-negative orthant, the simplex method will identify a vertex solution to the problem \eqref{eq:2nd-stage-dual}, even though the problem may admit unbounded solutions. 
Note that the polyhedral $\mathcal{P}_i \coloneqq \{v_i \in \rset{n_i + \ell_i}{} \mid B_i^Tv_i \leq q_i\}$ only admits a finite number of vertices $\mathcal{V}_i \coloneqq \{V_1, V_2, \ldots, V_{M}\}$ ($-\infty$ excluded). 
Thus,
$\mathcal{Q}_i(x_i; d_i) \coloneqq \max_{V_j \in \mathcal{V}_i} V_j^T \cdot [x_i; d_i]$,
which completes the proof that $\mathcal{Q}_i(x_i; d_i)$ is a piecewise linear function in $x_i$. 
It follows that the expected value function $\expt{D_i}{\mathcal{Q}_i(x_i; D_i)}$ is a convex function in $x_i$ \cite[Sec.~3.2.1]{boyd2004convex}. 
Applying the arguments in \cite[Sec.~V]{huang2021distributed} to the remaining parts of $\mathbb{J}_i(x_i; x_{-i})$, we can show that the pseudogradient $\mathbb{F}$ is strongly monotone. 
By \cite[Prop.~12.11]{palomar2010convex}, this multi-product assembly problem admits a unique Nash equilibrium. 
It can also be checked numerically that there exists a $\rho_{\mu} > 0$ such that the operator $\mathcal{R}^T\extgjacob + \frac{\rho_\mu}{2}L_n$ is maximally monotone. 
These arguments guarantee that Assumptions~\ref{asp:subgrad}, \ref{asp:ne-exist} and \ref{asp:convg} hold for this SGNEP. 


To guarantee that Assumption~\ref{asp:proj-stoch-subgrad} holds, it suffices to verify that the nonsmooth parts of the objectives fulfill these conditions. 
We can establish the interchangeability of subdifferential and integral using \cite[Thm.~7.52]{shapiro2014lectures}. 
We then consider the function $\phi_i(\chi) \coloneqq \max_{v_i \in \mathcal{P}_i}(v_i^T \cdot \chi)$, where $\mathcal{P}_i \coloneqq \{v_i \in \rset{n_i + \ell_i}{} \mid B_i^Tv_i \leq q_i\}$. 
Since the set $\mathcal{P}_i$ is nonempty, $\phi_i(\chi)$ is the support function of $\mathcal{P}_i$. 
By definition, the support function $\phi_i(\chi)$ is the conjugate function of the indicator function $\iota_{\mathcal{P}_i}(\chi)$, i.e., $\phi_i(\chi) = \max_{v_i \in \mathcal{P}_i}(v_i^T \cdot \chi) = \max_{v_i}(v_i^T \cdot \chi - \iota_{\mathcal{P}_i}(v_i))$. 
Since the set $\mathcal{P}_i$ is convex and closed, the function $\iota_{\mathcal{P}_i}(\chi)$ is convex, lower semicontinuous and proper.  
By \cite[Thm.~7.5 and (7.24)]{shapiro2014lectures}, we obtain $\partial \phi_i(\chi) = \argmax_{v_i} \{v_i^T\cdot \chi - \iota_{\mathcal{P}_i}(v_i)\} = \argmax_{v_i \in \mathcal{P}_i}\{v_i^T \cdot \chi\}$. 
Moreover, by the chain rule, the subdifferential should be
$\partial \mathcal{Q}_i(x_i; d_i) = \tilde{I}^T_i \cdot \argmax_{v_i \in \mathcal{P}_i} \{(\tilde{I}_i x_i + \tilde{d}_i)^T \cdot v_i\}$. 
As we discussed in the verification of Assumption~\ref{asp:subgrad}, the solution set of $\argmax_{v_i \in \mathcal{P}_i} \{(\tilde{I}_i x_i + \tilde{d}_i)^T \cdot v_i\}$ must contain at least one of $\mathcal{P}_i$'s vertices. 
Hence, we can always find a bounded subgradient of $\mathcal{Q}_i$ such that Assumption~\ref{asp:proj-stoch-subgrad} (ii) holds.

\subsubsection{Simulation Results}

We restrict each random variable $D_i$ to having a finite range $\{d_1, \ldots, d_L\}$ with the probability distribution $\{P_1, \ldots, P_{L}\}$. 
Under this restriction, the objective function of each player $i$ can be explicitly written as:
$\mathbb{J}_i(x_i; x_{-i}) = \frac{1}{2}x_i^TQ_ix_i + (\expt{}{C} + \Sigma Ax)^TA_ix_i + \sum_{l=1}^{L}P_l\mathcal{Q}_i(x_i; d_l)$.
The method proposed in \cite{huang2021distributed} can then be applied to compute the unique v-SGNE for reference. 
The performance of Algorithm~\ref{alg:node-edge} when solving this multi-product assembly problem is illustrated in Fig.~\ref{fig:two-stage}. 
The thin lines reflect the simple moving averages of these metrics with a window size of $20$. 
The curves of $T^{(k)} \propto k^{2.1}$ illustrate a steady convergence towards the v-GNE as suggested in Theorem~\ref{thm:summb}, while the trajectories of $T^{(k)} = 20$ stop decreasing after some iterations. 
The curves of $T^{(k)} \propto k$ also keep descending yet with a gentler trend compared with those of $T^{(k)} \propto k^{2.1}$, which suggests the possibility of some relaxations to the current conditions in Theorems~\ref{thm:main-convg-thm} and \ref{thm:summb}. 
For the detailed figure descriptions, please refer to Sec.~\ref{subsubsect:nash-dist-result}. 

\begin{figure}
    \centering
    \includegraphics[width=0.42\textwidth]{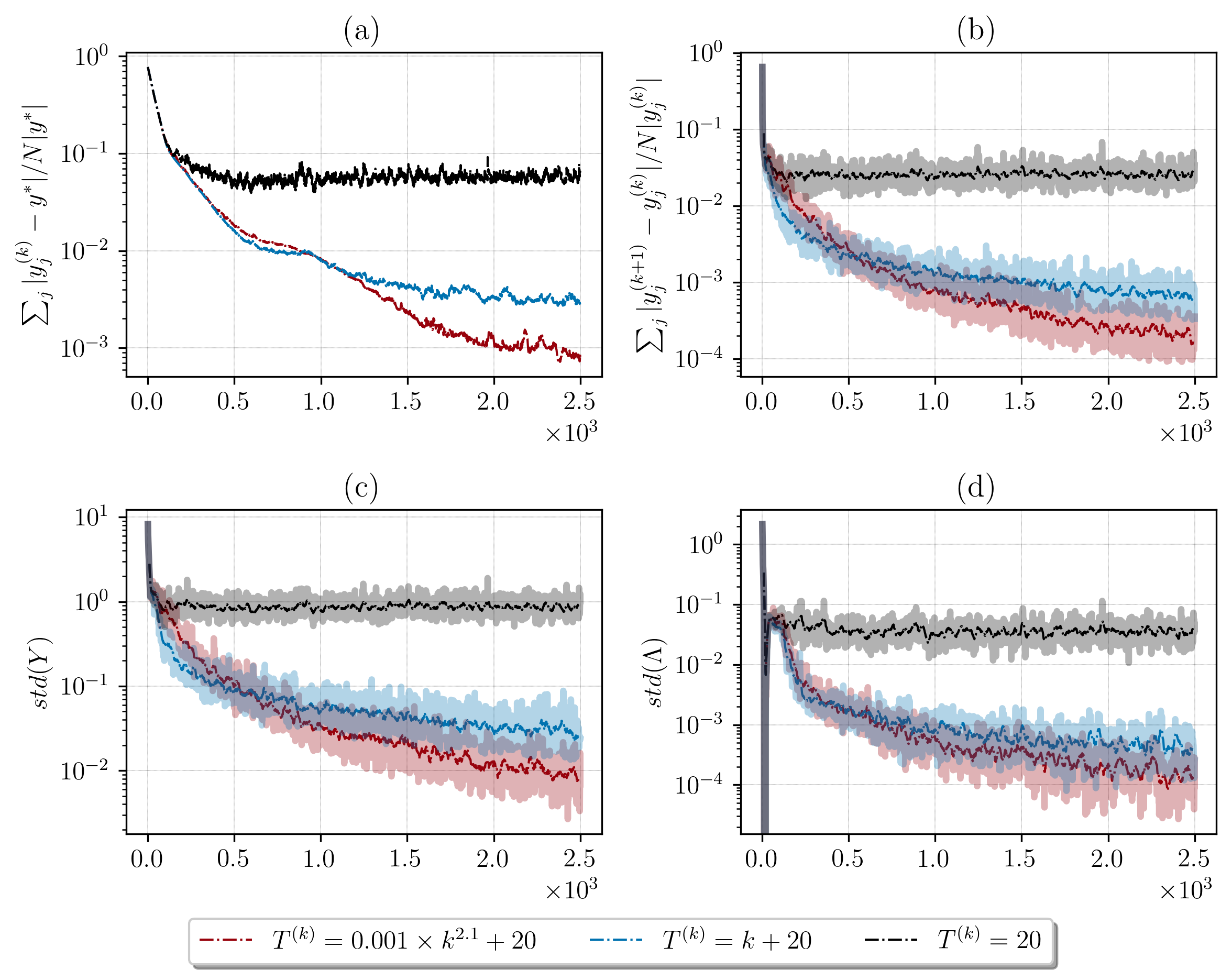}
    \caption{Performances of Alg.~\ref{alg:node-edge} for the Two-Stage Model}
    \label{fig:two-stage}
\end{figure}

\section{Conclusion and Future Directions}\label{sect:conclu}

In this paper, we study the stochastic generalized Nash equilibrium problem and propose a distributed stochastic algorithm under the partial-decision information setting based on solving augmented best-response subproblems induced by the Douglas-Rachford scheme. 
The proposed algorithm is proved to converge to a true variational stochastic generalized Nash equilibrium if the sequence of inertial step sizes and the inverse of the number of realizations per major iteration decrease altogether at a proper rate. 
This raises the question if there exists a less conservative bound for this decreasing rate such that the proposed algorithm can still converge yet with a faster convergence rate and fewer observations needed per major iteration. 
Another interesting work remains concerning the convergence rate analysis of the proposed algorithm. 
As we have previously mentioned, the fixed point iteration discussed in this paper engages two reflected resolvent operators, which merely admit quasinonexpansiveness rather than contractiveness. 
The convergence rate analysis under this setting remains an under-explored yet increasingly active direction \cite{boley2013local, liang2018local, cui2020variance}. 
Finally, although we only analyze the projected stochastic subgradient method, the main convergence result in Theorem~\ref{thm:main-convg-thm} actually allows a lot of possibilities. 
It would be interesting to develop inexact solvers based on different stochastic optimization schemes, e.g. proximal point methods \cite{asi2019importance}, that are more sample-efficient or further relax the assumptions made.

\appendices
\section*{Appendix}
\renewcommand{\thesubsection}{\Alph{subsection}}

\subsection{Proof of Theorem \ref{thm:main-convg-thm}}\label{pf:main-convg-thm}

\begin{proof}
The following proof is largely inspired from that of \cite[Prop.~5.34]{BauschkeHeinzH2017CAaM} for deterministic sequences and nonexpansive operators case with suitable modifications. 
Given an arbitrary initial point $\tilde{\psi}^{(0)} = \tilde{\psi}^{(0)}_*$, we let $(\tilde{\psi}^{(k)})_{k \in \nset{}{}}$ denote the sequence generated by the approximate iteration \eqref{eq:approx-iter}.
Another auxiliary sequence $(\tilde{\psi}^{(k)}_*)_{k \in \nset{}{}}$ is constructed by letting $\tilde{\psi}^{(k+1)}_* \coloneqq \mathscr{P}_*(\tilde{\psi}^{(k)})$. 
We next try to extract a recursive relationship w.r.t. $\norm{\tilde{\psi}^{(k+1)} - \tilde{\psi}^*}^2_\pspace$ to establish that $\sum_{k \in \nset{}{}}\res{\tilde{\psi}^{(k)}} < \infty$ a.s., where $\tilde{\psi}^*$ is a fixed point of $\mathscr{R}_*$. 
Writing the explicit updating formula of $\norm{\tilde{\psi}^{(k+1)}_* - \tilde{\psi}^*}^2_\pspace$ yields
\begin{align*}\label{eq:main-covg-exact-updt}
& \norm{\supsub{\tilde{\psi}}{(k+1)}{*} - \supsub{\tilde{\psi}}{*}{}}^2_\pspace = \norm{(1 - \supsub{\gamma}{(k)}{})\supsub{\tilde{\psi}}{(k)}{} + \supsub{\gamma}{(k)}{}\mathscr{R}_*(\supsub{\tilde{\psi}}{(k)}{}) - \supsub{\tilde{\psi}}{*}{}}^2_\pspace \\
&  = (1 - \supsub{\gamma}{(k)}{})\norm{\supsub{\tilde{\psi}}{(k)}{} - \supsub{\tilde{\psi}}{*}{}}^2_\pspace + \supsub{\gamma}{(k)}{}\norm{\mathscr{R}_*(\supsub{\tilde{\psi}}{(k)}{}) - \mathscr{R}_*(\supsub{\tilde{\psi}}{*}{})}^2_\pspace \\
& \qquad - \supsub{\gamma}{(k)}{}(1 - \supsub{\gamma}{(k)}{})\norm{\mathscr{R}_*(\supsub{\tilde{\psi}}{(k)}{}) - \supsub{\tilde{\psi}}{(k)}{}}^2_\pspace \\ 
& \leq \norm{\supsub{\tilde{\psi}}{(k)}{} - \supsub{\tilde{\psi}}{*}{}}^2_\pspace - \supsub{\gamma}{(k)}{}(1 - \supsub{\gamma}{(k)}{})(\res{\supsub{\tilde{\psi}}{(k)}{}})^2.
\end{align*}
where the inequality follows from the fact that $\mathscr{R}_*$ is quasinonexpansive. 
Next, we derive a recursive relationship for $\norm{\tilde{\psi}^{(k+1)} - \tilde{\psi}^*}^2_\pspace$ as follows:
\begin{align*}
& \norm{\tilde{\psi}^{(k+1)} - \tilde{\psi}^*}^2_\pspace = \norm{\gamma^{(k)}\epsilon^{(k)} + \tilde{\psi}^{(k+1)}_* - \tilde{\psi}^*}^2_\pspace \\
& = \norm{\tilde{\psi}^{(k+1)}_* - \tilde{\psi}^*}^2_\pspace  + (\gamma^{(k)}\varepsilon^{(k)})^2  + 2\langle \gamma^{(k)}\epsilon^{(k)}, \tilde{\psi}^{(k+1)}_* - \tilde{\psi}^*\rangle_\pspace \\
& \leq \norm{\tilde{\psi}^{(k)} - \tilde{\psi}^*}^2_\pspace - \gamma^{(k)}(1 - \gamma^{(k)})(\res{\tilde{\psi}^{(k)}})^2  \\
& \qquad + (\gamma^{(k)}\varepsilon^{(k)})^2 + 2\gamma^{(k)}\varepsilon^{(k)}\norm{\tilde{\psi}^{(k)} - \tilde{\psi}^*}_\pspace ,
\end{align*}
where the last inequality follows from the relation derived above and the Cauchy-Schwarz inequality. 
Taking conditional expectation $\expt{}{\cdot \mid \mathcal{F}^{(k)}}$ on both sides yields:
\begin{equation}\label{eq:main-covg-rs-recur}
\begin{split}
& \expt{}{\norm{\tilde{\psi}^{(k+1)} - \tilde{\psi}^*}^2_\pspace \mid \mathcal{F}^{(k)}} \\
& \leq \norm{\tilde{\psi}^{(k)} - \tilde{\psi}^*}^2_\pspace - \gamma^{(k)}(1 - \gamma^{(k)})(\res{\tilde{\psi}^{(k)}})^2 \\
& \quad + \expt{}{2\gamma^{(k)}\varepsilon^{(k)}\norm{\tilde{\psi}^{(k)} - \tilde{\psi}^*}_\pspace + (\gamma^{(k)}\varepsilon^{(k)})^2\mid \mathcal{F}^{(k)}}. 
\end{split}
\end{equation}
Based on the fact that $\sum_{k \in \nset{}{}} \gamma^{(k)}\expt{}{\varepsilon^{(k)} \mid \mathcal{F}^{(k)}} < +\infty$ a.s. and $(\norm{\tilde{\psi}^{(k)} - \tilde{\psi}^*}_\pspace)_{k \in \nset{}{}}$ is bounded a.s., we can obtain that $\sum_{k \in \nset{}{}} \expt{}{2\gamma^{(k)}\varepsilon^{(k)}\norm{\tilde{\psi}^{(k)} - \tilde{\psi}^*}_\pspace + (\gamma^{(k)}\varepsilon^{(k)})^2\mid \mathcal{F}^{(k)}} < \infty$ a.s. 
By applying the Robbins-Siegmund theorem, we can then conclude that on a set $\hat{\Omega}$ which has probability one, 
$\sum_{k \in \nset{}{}}\gamma^{(k)}(1 - \gamma^{(k)})\res{\tilde{\psi}^{(k)}}^2 < \infty$ with $\gamma^{(k)} \in [0, 1]$ and $\sum_{k \in \nset{}{}}\gamma^{(k)}(1 - \gamma^{(k)}) = +\infty$. 
Now we fix an arbitrary sample path $\hat{\omega} \in \hat{\Omega}$ for subsequent analysis, while omitting $\hat{\omega}$ for brevity. 
In the following we will prove by contradiction that $\liminf_{k \to \infty} \res{\tilde{\psi}^{(k)}}^2 = 0$
Suppose otherwise that $\liminf_{k \to \infty} \res{\tilde{\psi}^{(k)}}^2 = \delta$, where $\delta > 0$ is some positive constant.  
Then there exists a sufficiently large $K_\delta$ such that for any $k > K_\delta$, $\res{\tilde{\psi}^{(k)}}^2 > \delta/2$. 
By this lower bound, we have $\sum_{k > K_\delta} \gamma^{(k)}(1 - \gamma^{(k)})\res{\tilde{\psi}^{(k)}}^2 > \delta/2 \sum_{k > K_\delta} \gamma^{(k)}(1 - \gamma^{(k)}) = +\infty$, which contradicts the previous statement that $\sum_{k \in \nset{}{}}\gamma^{(k)}(1 - \gamma^{(k)})\res{\tilde{\psi}^{(k)}}^2 < \infty$. 
This shows that $\liminf_{k \to \infty} \res{\tilde{\psi}^{(k)}}^2 = 0$. 
As a result, there exists a subsequence, denoted by $(\tilde{\psi}^{(k_i)})_{i \in \nset{}{}}$, such that $\lim_{i \to \infty} \res{\tilde{\psi}^{(k_i)}}^2 = 0$.

Moreover, the above subsequence $(\tilde{\psi}^{(k_i)})_{k_i \in \nset{}{}}$ is bounded and thus has a convergent subsequence $(\tilde{\psi}^{(l_i)})_{i \in \nset{}{}}$ where $(l_i)_{i \in \nset{}{}} \subseteq (k_i)_{i \in \nset{}{}}$ such that $\lim_{i \to \infty} \tilde{\psi}^{(l_i)} = \tilde{\psi}^\dagger$. 
If Assumption~\ref{asp:convg}(i) holds, by definition, $\mathscr{R}_*$ is a nonexpansive mapping. 
It then follows from \cite[Cor.~4.28]{BauschkeHeinzH2017CAaM} that $\tilde{\psi}^\dagger \in \fix{\mathscr{R}_*}$. 
If Assumption~\ref{asp:convg}(ii) holds instead, from \cite[Lemma~6]{huang2021distributed},  $\rrcomp_*$ is a continuous mapping, i.e., $\lim_{i \to \infty} \res{\tilde{\psi}^{(l_i)}} = 0$ implies $\rrcomp_*(\tilde{\psi}^\dagger) = \tilde{\psi}^\dagger$ and hence $\tilde{\psi}^\dagger \in \fix{\mathscr{R}_*}$. 
Therefore we can substitute $\tilde{\psi}^*$ in \eqref{eq:main-covg-rs-recur} with $\tilde{\psi}^\dagger$. 
\Tblue{By \cite[Thm.~1]{robbins1971convergence}}, $\lim_{k \to \infty} \norm{\tilde{\psi}^{(k)} - \tilde{\psi}^\dagger}^2_\pspace$ exists.
Since $(\tilde{\psi}^{(l_i)})_{i \in \nset{}{}}$ is a subsubsequence of $(\tilde{\psi}^{(k)})_{k \in \nset{}{}}$ converging to the fixed point $\tilde{\psi}^\dagger$, we can conclude that $\lim_{k \to \infty} \norm{\tilde{\psi}^{(k)} - \tilde{\psi}^\dagger}^2_\pspace = 0$, and hence $\lim_{k \to \infty}\tilde{\psi}^{(k)} = \tilde{\psi}^\dagger$.
Altogether, $\psi^\dagger \coloneqq J_{\poptA}(\tilde{\psi}^\dagger)$ belongs to the zero set of $\optT$ in \eqref{eq:optT}.
Combining this with the conclusions of Theorem~\ref{thm:zerokkt}, the proof is complete. 
\end{proof}

\subsection{Proof of Lemma \ref{le:proj-stoch-convg-rate}}
\label{pf:proj-stoch-convg-rate}
\begin{proof}
For each player $i \in \playerN$, at an arbitrary major iteration $k$ and its minor iteration $t = 0, \ldots, T^{(k)}_{i} - 1$, by applying the update inside the for-loop of Algorithm~\ref{alg:proj-stoch-subgrad} and using the nonexpansiveness of the projection operator onto a convex set, we can obtain the following inequality of the distance between the approximate minimizer after the $t$th minor iteration $y^{i(k+1)}_{i,t+1}$ and the accurate minimizer $y^{i(k+1)}_{i,*}$:
\begin{equation}\label{le:proj-grad-1}
\norm{y^{i(k+1)}_{i,t+1} - y^{i(k+1)}_{i,*}}^2_2 \leq \norm{y^{i(k+1)}_{i,t} - \kappa_{i,t}\cdot g^{(k)}_{i,t} - y^{i(k+1)}_{i,*}}^2_2.
\end{equation}
Expanding the $\ell^2$ norm and taking conditional expectation $\bexpt{}{\cdot \mid \sigma\{\mathcal{F}_k, \xi^{(k)}_{i, [t]}\}}$ on both sides of \eqref{le:proj-grad-1} yields:
\begin{align}\label{eq:proj-subgrad-1st}
\begin{split}
& \bexpt{}{\norm{y^{i(k+1)}_{i, t+1} - y^{i(k+1)}_{i, *}}^2_2 \mid \sigma\{\mathcal{F}_k, \xi^{(k)}_{i, [t]}\}} \\
&  \leq \kappa_{i,t}^2\bexpt{}{\norm{g^{(k)}_{i,t}}^2_2 \mid \sigma\{\mathcal{F}_k, \xi^{(k)}_{i, [t]}\}} + \norm{y^{i(k+1)}_{i,t} - y^{i(k+1)}_{i,*}}^2_2 \\
& \qquad - 2\kappa_{i,t}\langle y^{i(k+1)}_{i,t} - y^{i(k+1)}_{i,*}, \mathbbm{g}^{(k)}_{i,t}\rangle , \\
\end{split}
\end{align}
where $\mathbbm{g}^{(k)}_{i,t} \coloneqq \expt{}{g^{(k)}_{i,t} \mid \sigma\{\mathcal{F}_k, \xi^{(k+1)}_{i, [t]}\}} \in \partial_{y^i_i}\hat{\mathbb{J}}^{(k)}_i(y^{i(k+1)}_{i, t})$ by Assumption~\ref{asp:proj-stoch-subgrad}. 
Using the $\frac{1}{\tau_{1i}}$-strong convexity of $\hat{\mathbb{J}}^{(k)}_{i}$, the inner product inside the inequality \eqref{eq:proj-subgrad-1st} satisfies
$\langle y^{i(k+1)}_{i,t} - y^{i(k+1)}_{i,*}, \mathbbm{g}^{(k)}_{i,t}\rangle \geq \hat{\mathbb{J}}^{(k)}_{i}(y^{i(k+1)}_{i,t}) - \hat{\mathbb{J}}^{(k)}_{i}(y^{i(k+1)}_{i,*}) + \frac{1}{2\tau_{1i}}\norm{y^{i(k+1)}_{i,t} - y^{i(k+1)}_{i,*}}^2_2$. 
We then take conditional expectations $\expt{}{\cdot \mid \mathcal{F}_k}$ on both sides of the above inequality.
By the rule of successive conditioning and the fact that $y^{i(k+1)}_{i,*}$ minimizes $\hat{\mathbb{J}}^{(k)}_{i}(\cdot)$, the following inequality holds a.s.:
\begin{equation}\label{eq:proj-subgrad-3rd}
\begin{split}
& \bexpt{}{\norm{y^{i(k+1)}_{i, t+1} - y^{i(k+1)}_{i,*}}^2_2 \mid \mathcal{F}_k} \leq \kappa_{i,t}^2\bexpt{}{\norm{g^{(k)}_{i,t}}^2_2 \mid \mathcal{F}_k}  \\
& \quad + (1 - \frac{\kappa_{i,t}}{\tau_{1i}})\bexpt{}{\norm{y^{i(k+1)}_{i, t} - y^{i(k+1)}_{i,*}}^2_2 \mid \mathcal{F}_k}.
\end{split}
\end{equation}
Re-arranging \eqref{eq:proj-subgrad-3rd} and applying Assumption~\ref{asp:proj-stoch-subgrad}(ii), we have \Tblue{that} the following holds a.s.:
\begin{equation}\label{eq:tele-term}
\begin{split}
&  \frac{1}{\kappa_{i,t}}\bexpt{}{\norm{y^{i(k+1)}_{i,t+1} - y^{i(k+1)}_{i,*}}^2_2 \mid \mathcal{F}_k}- (\frac{1}{\kappa_{i,t}} - \frac{1}{\tau_{1i}})\cdot  \\
& \quad \bexpt{}{\norm{y^{i(k+1)}_{i,t} - y^{i(k+1)}_{i,*}}^2_2 \mid \mathcal{F}_k} \leq \kappa_{i,t}(\alpha_{g,i}^2\norm{\tilde{\psi}^{(k)}}^2_2 + \beta^2_{g,i}).
\end{split}
\end{equation}
By setting $\kappa_{i,t} \coloneqq \frac{2\tau_{1i}}{t+2}$, multiplying both sides of \eqref{eq:tele-term} by $(t+1)/2$, and summing \eqref{eq:tele-term} for $t=0, \ldots, T-1$, for an arbitrary $T \in \{1, \ldots, T^{(k)}_{i}\}$, we obtain a telescoping sum and have that the following holds a.s.:
\begin{equation}\label{eq:proj-subgrad-4th}
\begin{split}
& \frac{(T+1)T}{4\tau_{1i}} \bexpt{}{\norm{y^{i(k+1)}_{i,T} - y^{i(k+1)}_{i,*}}^2_2 \mid \mathcal{F}_k} \\
& \qquad\qquad \leq \sum_{t=0}^{T-1}\frac{t+1}{2}\cdot\frac{2\tau_{1i}}{t+2}(\alpha_{g,i}^2\norm{\tilde{\psi}^{(k)}}^2_2 + \beta^2_{g,i}).
\end{split}
\end{equation}
Simplifying \eqref{eq:proj-subgrad-4th}, we deduce that
$\bexpt{}{\norm{y^{i(k+1)}_{i,T} - y^{i(k+1)}_{i,*}}^2_2 \mid \mathcal{F}_k} \leq 4\tau_{1i}^2T^{-1}(\alpha^2_{g,i}\norm{\tilde{\psi}^{(k)}}^2_2 + \beta^2_{g,i})$ a.s.
\end{proof}

\subsection{Proof of Lemma \ref{le:convg-rate-augvec}}
\label{pf:convg-rate-augvec}
\begin{proof}
By the nonexpansiveness of the reflected resolvent $R_{\poptB}$, the approximate error $\varepsilon^{(k)}$ should satisfy:
\begin{equation}
\begin{split}
\expt{}{\varepsilon^{(k)} \mid \mathcal{F}_k} &\leq \bexpt{}{\norm{R_{\appoptA}(\tilde{\psi}^{(k)}) - R_{\poptA}(\tilde{\psi}^{(k)})}_\pspace \mid \mathcal{F}_k} \\
& = 2\bexpt{}{\norm{\psi^{(k+1)} - \psi^{(k+1)}_*}_\pspace \mid \mathcal{F}_k},
\end{split}
\end{equation}
where $\psi^{(k+1)} = [\by^{(k+1)}; \blambda^{(k+1)}; \bmu^{(k+1)}; \bz^{(k+1)}] \coloneqq J_{\appoptA}(\tilde{\psi}^{(k)})$ is the stack vector obtained by using the inexact solver suggested in Algorithm~\ref{alg:proj-stoch-subgrad} and $\psi^{(k+1)}_* = [\by^{(k+1)}_*; \blambda^{(k+1)}_*; \bmu^{(k+1)}_*; \bz^{(k+1)}_*] \coloneqq J_{\poptA}(\tilde{\psi}^{(k)})$ is the one using the accurate solver. 
Given the conclusion of Lemma~\ref{le:proj-stoch-convg-rate} and the first two for-loops in Algorithm~\ref{alg:node-edge}, the approximate error of the dual variables $\blambda$ has the following upper bound:
\begin{align*}
& \bexpt{}{\norm{\blambda^{(k+1)} - \blambda^{(k+1)}_*}_2 \mid \mathcal{F}_k}\leq \bexpt{}{\norm{\tau_2\Lambda\mathcal{R}(\by^{(k+1)} - \by^{(k+1)}_*)}_2 \mid \mathcal{F}_k} \\
& \qquad \leq \norm{\tau_{2}\Lambda\mathcal{R}}_2 \cdot \bexpt{}{\norm{\by^{(k+1)} - \by^{(k+1)}_*}_2 \mid \mathcal{F}_k}. 
\end{align*}
Similar results can be trivially derived for $\bmu$ and $\bz$, the details of which are omitted for brevity.
Altogether, we have that the following relation $\bexpt{}{\norm{\psi^{(k+1)} - \psi^{(k+1)}_*}_2 \mid \mathcal{F}_k} \leq C_1 \cdot \bexpt{}{\norm{\by^{(k+1)} - \by^{(k+1)}_*}_2 \mid \mathcal{F}_k}$ holds for some constant $C_1$. 
For each player $i \in \playerN$, the local estimates of others' decisions are the same in $y^{i(k+1)}$ and $y^{i(k+1)}_*$, while the local decisions, by Lemma~\ref{le:proj-stoch-convg-rate}, satisfy $\bexpt{}{\norm{y^{i(k+1)}_i - y^{i(k+1)}_{i,*}}^2_2 \mid \mathcal{F}_k} \leq 4\tau_{1i}^2(T^{(k)}_i)^{-1}(\alpha^2_{g,i}\norm{\tilde{\psi}^{(k)}}^2_2 + \beta^2_{g,i})$ a.s. for each $i \in \playerN$. 
Picking the maximum coefficients $\bar{\alpha}_{g} \coloneqq \max\{\alpha_{g,i}: i \in \playerN\}$, $\bar{\beta}_{g} \coloneqq \max\{\beta_{g,i}: i \in \playerN\}$, $\bar{\tau}_{1} \coloneqq \max\{\tau_{1, i}: i \in \playerN\}$ and the minimum minor steps taken $\ubar{T}^{(k)} \coloneqq \min\{T^{(k)}_{i}: i \in \playerN\}$ over all players. 
By Jensen's inequality and the non-negativity of $\alpha_{g,i}$, $\beta_{g,i}$, and $\norm{\tilde{\psi}^{(k)}}$, an upper bound for the stacked local decisions and estimates is given by:
\begin{equation*}
\begin{split}
& \expt{}{\norm{\by^{(k+1)} - \by^{(k+1)}_*}_2 \mid \mathcal{F}_k} 
\leq (\expt{}{\sum_{i \in \playerN}\norm{y^{i(k+1)}_i - y^{i(k+1)}_{i*}}^2_2 \mid \mathcal{F}_k})^{1/2} \\
& \qquad \leq 2\sqrt{N}\bar{\tau}_{1}(\ubar{T}^{(k)})^{-1/2}(\bar{\alpha}_{g}\norm{\tilde{\psi}^{(k)}}_2 + \bar{\beta}_{g}), \text{ a.s.}
\end{split}
\end{equation*}
Combining the above inequalities, we derive the following a.s. upper bound in the Euclidean space:
\begin{equation}\label{eq:augvec-eucl}
\bexpt{}{\norm{\psi^{(k+1)} - \psi^{(k+1)}_*}_2 \mid \mathcal{F}_k} \leq C_2(\ubar{T}^{(k)})^{-1/2}(\bar{\alpha}_{g}\norm{\tilde{\psi}^{(k)}}_2 + \bar{\beta}_{g}),
\end{equation}
where $C_2 \coloneqq 2C_1\bar{\tau}_{1}\sqrt{N}$. 
We convert the above conclusion from the Euclidean space to the inner product space $\pspace$ defined by the positive definite design matrix $\Phi$. 
The maximum (resp. minimum) eigenvalue of $\Phi$ is denoted by $\bar{\sigma}_{\Phi}$ (resp. $\ubar{\sigma}_{\Phi}$). 
Then \eqref{eq:augvec-eucl} implies the following relation holds a.s. in $\pspace$:
\begin{equation}
\bexpt{}{\norm{\psi^{(k+1)} - \psi^{(k+1)}_*}_\pspace \mid \mathcal{F}_k}  \leq \frac{C_2\sqrt{\bar{\sigma}_{\Phi}}}{(\ubar{T}^{(k)})^{1/2}}(\frac{\bar{\alpha}_{g}}{\sqrt{\ubar{\sigma}_{\Phi}}}\norm{\tilde{\psi}^{(k)}}_\pspace + \bar{\beta}_{g}). 
\end{equation}
Hence, there exist positive constants $\alpha_{\psi}$ and $\beta_{\psi}$ independent of $k$ such that 
$\expt{}{\varepsilon^{(k)} \mid \mathcal{F}_k} \leq (\ubar{T}^{(k)})^{-1/2}(\alpha_\psi\norm{\tilde{\psi}^{(k)}}_\pspace + \beta_{\psi})$ a.s.
\end{proof}

\subsection{Proof of Theorem \ref{thm:summb}}
\label{pf:summb}
\begin{proof}
Consider a sequence of augmented vectors $(\tilde{\psi}^{(k)})_{k \in \nset{}{}}$ generated by the approximate iteration $\mathscr{P} = \idty + \gamma^{(k)}(\rrcomp - \idty)$ and a sequence $(\tilde{\psi}^{(k)}_*)_{k \in \nset{}{}}$ generated by $\tilde{\psi}^{(k+1)}_{*} \coloneqq \mathscr{P}_*(\tilde{\psi}^{(k)})$. 
Let $\tilde{\psi}^*$ denote one of the fixed points of $\mathscr{R}_*$. 
To prove that $(\tilde{\psi}^{(k)})_{k \in \nset{}{}}$ is bounded a.s., note that
\begin{align*}
& \expt{}{\norm{\tilde{\psi}^{(k+1)} - \tilde{\psi}^*}_\pspace \mid \mathcal{F}_k} = \expt{}{\norm{\tilde{\psi}^{(k+1)} - \tilde{\psi}^{(k+1)}_* + \tilde{\psi}^{(k+1)}_* - \tilde{\psi}^*}_\pspace \mid \mathcal{F}_k} \\
& \leq \gamma^{(k)}\expt{}{\varepsilon^{(k)} \mid \mathcal{F}_k} + \expt{}{\norm{\mathscr{P}_*(\tilde{\psi}^{(k)}) - \mathscr{P}_*(\tilde{\psi}^*)}_{\pspace} \mid \mathcal{F}_k}. 
\end{align*}
Let $\gamma^{(k)}_T \coloneqq \gamma^{(k)}(\ubar{T}^{(k)})^{-1/2}$. 
By applying Lemma~\ref{le:convg-rate-augvec} and using the fact that $\mathscr{P}_*$ is (quasi)nonexpansive, we have:
\begin{align*}
& \expt{}{\norm{\tilde{\psi}^{(k+1)} - \tilde{\psi}^*}_\pspace \mid \mathcal{F}_k} \\
& \leq  \gamma^{(k)}_T(\alpha_\psi \norm{\tilde{\psi}^{(k)}}_\pspace + \beta_\psi) + \expt{}{\norm{\tilde{\psi}^{(k)} - \tilde{\psi}^*}_\pspace \mid \mathcal{F}_k} \\
& = \gamma^{(k)}_T(\alpha_\psi \norm{\tilde{\psi}^{(k)} - \tilde{\psi}^* + \tilde{\psi}^*}_\pspace + \beta_\psi) + \norm{\tilde{\psi}^{(k)} - \tilde{\psi}^*}_\pspace \\ 
& \leq (1 + \alpha_\psi \gamma^{(k)}_{T})\norm{\tilde{\psi}^{(k)} - \tilde{\psi}^*}_\pspace + \gamma^{(k)}_{T}(\alpha_\psi\norm{\tilde{\psi}^*}_\pspace + \beta_\psi), \text{a.s.} 
\end{align*}
Since $\norm{\tilde{\psi}^*}_\pspace < \infty$ and we assume that $(\gamma^{(k)}_T)_{k \in \nset{}{}}$ is a summable sequence, \Tblue{the Robbins-Siegmund Theorem (\cite[Thm.~1]{robbins1971convergence})} can be applied to show $\lim_{k \to \infty} \norm{\tilde{\psi}^{(k)} - \tilde{\psi}^*}_\pspace$ exists and is finite a.s.
Consequently, there exists a set $\hat{\Omega}$ which has probability one, such that for any $\hat{\omega} \in \hat{\Omega}$, the sequence $(\norm{\tilde{\psi}^{(k)}(\hat{\omega}) - \tilde{\psi}^*}_\pspace)_{k \in \nset{}{}}$ is bounded. 
Therefore, we can find some constant $B(\hat{\omega})$ which satisfies, for all $k \in \nset{}{}$, $\norm{\tilde{\psi}^{(k)}(\hat{\omega})}_\pspace = \norm{\tilde{\psi}^{(k)}(\hat{\omega}) - \tilde{\psi}^* + \tilde{\psi}^*}_\pspace \leq \norm{\tilde{\psi}^{(k)}(\hat{\omega}) - \tilde{\psi}^*}_\pspace + \norm{\tilde{\psi}^*}_\pspace \leq B(\hat{\omega})$. 

Since the deterministic sequence $(\norm{\tilde{\psi}^{(k)}(\hat{\omega})}_\pspace)_{k \in \nset{}{}}$ is upper bounded by a constant $B(\hat{\omega})$ for any $\hat{\omega} \in \hat{\Omega}$, 
combining Lemma~\ref{le:convg-rate-augvec} and the summability of $(\gamma^{(k)}_T)_{k \in \nset{}{}}$, 
we finally can conclude that
$\sum_{k \in \nset{}{}}\gamma^{(k)}\expt{}{\varepsilon^{(k)} \mid \mathcal{F}_k}(\hat{\omega}) \leq \sum_{k \in \nset{}{}}\gamma^{(k)}_T(\alpha_{\psi}\norm{\tilde{\psi}^{(k)}(\hat{\omega})}_\pspace + \beta_{\psi})
\leq \sum_{k \in \nset{}{}}\gamma^{(k)}_T(\alpha_\psi B(\hat{\omega}) + \beta_\psi) < \infty$.
\end{proof}

\subsection{Almost-Sure Convergence of Subroutine~\ref{alg:proj-stoch-subgrad}}

We now let $J^{(k)}_{\appoptA}$ denote the (scenario-based) approximate operator for the exact resolvent $J_{\poptA}$ at the $k$-th iteration. 
As a reminder, note that the explicit steps of $J^{(k)}_{\appoptA}$ are presented in the first player and edge loops in Algorithm~\ref{alg:node-edge}, before we implement the reflected steps. 
In the following lemma, we are going to establish that the result of $J^{(k)}_{\appoptA}$ can approximate that of $J_{\poptA}$ with arbitrary accuracy almost surely when a sufficiently large number of stochastic subgradient steps have been taken. 

\begin{lemma}
Suppose Assumptions~\ref{asp:subgrad}, \ref{asp:fesb-set}, \ref{asp:phi-pd} , and \ref{asp:proj-stoch-subgrad} hold, and $\tilde{\psi}$ is an arbitrary bounded stack vector. 
In addition, let the number of subgradient steps taken per iteration satisfy $\lim_{k \to \infty} \ubar{T}^{(k)} = \infty$. 
Then $\lim_{k\to \infty} J^{(k)}_{\appoptA}(\tilde{\psi}) = J_{\poptA}(\tilde{\psi})$ a.s.
\end{lemma}

\begin{proof}
From the explicit updating steps presented in Algorithm~\ref{alg:node-edge}, it is straightforward that, in the resulting vectors, the entries associated with the local estimates $\{y^{-i}_i\}$ keep the same for $J^{(k)}_{\appoptA}$ and $J_{\poptA}$, and the entries associated with the dual variables $\blambda$, $\bmu$, and $\bz$ are some linear transformations of $\tilde{\psi}$ and those associated with $\{y^i_i\}$. 
Hence, it suffices to prove $\lim_{t \to \infty} \norm{y^i_{i,t} - y^i_{i,*}}^2_2 = 0$ a.s. for all $i \in \playerN$. 
We let $g_{i,t}$ denote the scenario-based gradient evaluated at the point $y^i_{i,t}$, and $g^*_{i,t}$ the gradient corresponding to the expected-value augmented objective. 
Start by noting that the distance between $y^i_{i,t}$ the point obtained after $t$ minor iteration steps and $y^i_{i,*}$ the minimizer of the expected-valued augmented objective $\hat{\mathbb{J}}^{(k)}_i$ satisfies:
\begin{align*}
& \norm{y^i_{i,t+1} - y^i_{i,*}}^2_2 = \norm{\proj_{\mathcal{X}^B_i}(y^i_{i,t} - \kappa_{i,t} g_{i,t}) - \proj_{\mathcal{X}^B_i}(y^i_{i,*})}^2_2 \\
& \qquad \leq \norm{y^i_{i,t} - \kappa_{i,t} g_{i,t} - y^i_{i,*}}^2_2 \\
& \qquad = \norm{y^i_{i,t}- y^i_{i,*}}^2_2 - 2\kappa_{i,t}\langle y^i_{i,t}- y^i_{i,*}, g_{i,t}\rangle + (\kappa_{i,t})^2\norm{g_{i,t}}^2_2 \\
& \qquad = \norm{y^i_{i,t}- y^i_{i,*}}^2_2 - 2\kappa_{i,t}\langle y^i_{i,t}- y^i_{i,*}, g^*_{i,t}\rangle \\
& \qquad\qquad + 2\langle y^i_{i,t}- y^i_{i,*}, g^*_{i,t} - g_{i,t} \rangle + (\kappa_{i,t})^2\norm{g_{i,t}}^2_2. 
\end{align*}
We then construct the following $\sigma$-field:
\begin{align*}
    \bar{\mathcal{F}}_t \coloneqq \sigma\{\xi_{i,0}, \ldots, \xi_{i,t-1}\}. 
\end{align*}
Taking the conditional expectation $\expt{}{\cdot \mid \bar{\mathcal{F}}_t}$ on both side of the above inequality yields:
\begin{align*}
& \expt{}{\norm{y^i_{i,t+1} - y^i_{i,*}}^2_2 \mid \bar{\mathcal{F}}_t} \\
& \qquad \leq  (1-\frac{\kappa_{i,t}}{\tau_{1i}})\norm{y^{i}_{i,t} - y^{i}_{i,*}}^2_2 + (\kappa_{i,t})^2\expt{}{\norm{g_{i,t}}^2_2 \mid \bar{\mathcal{F}}_t} \\
& \qquad \leq  (1-\frac{\kappa_{i,t}}{\tau_{1i}})\norm{y^{i}_{i,t} - y^{i}_{i,*}}^2_2 
+ \kappa_{i,t}^2(\alpha^2_{g,i}\norm{\tilde{\psi}}^2_2 + \beta^2_{g,i})
\end{align*}
where the first inequality follows from the strong convexity of $\hat{\mathbb{J}}^{(k)}_i$, and the second is based on Assumption~\ref{asp:proj-stoch-subgrad}(ii). 
By leveraging \cite[Lemma~2.2.10]{polyak1987introduction}, we can conclude that $\lim_{t \to \infty} \norm{y^{i}_{i,t} - y^{i}_{i,*}}^2_2 = 0$ a.s., which completes the proof.  
\end{proof}

\bibliographystyle{IEEEtran}
\bibliography{IEEEabrv,references}

\begin{thebibliography}{10}
\providecommand{\url}[1]{#1}
\csname url@samestyle\endcsname
\providecommand{\newblock}{\relax}
\providecommand{\bibinfo}[2]{#2}
\providecommand{\BIBentrySTDinterwordspacing}{\spaceskip=0pt\relax}
\providecommand{\BIBentryALTinterwordstretchfactor}{4}
\providecommand{\BIBentryALTinterwordspacing}{\spaceskip=\fontdimen2\font plus
\BIBentryALTinterwordstretchfactor\fontdimen3\font minus
  \fontdimen4\font\relax}
\providecommand{\BIBforeignlanguage}[2]{{%
\expandafter\ifx\csname l@#1\endcsname\relax
\typeout{** WARNING: IEEEtran.bst: No hyphenation pattern has been}%
\typeout{** loaded for the language `#1'. Using the pattern for}%
\typeout{** the default language instead.}%
\else
\language=\csname l@#1\endcsname
\fi
#2}}
\providecommand{\BIBdecl}{\relax}
\BIBdecl

\bibitem{huang2021sgnep}
Y.~Huang and J.~Hu, ``A distributed {Douglas-Rachford} based algorithm for
  stochastic {GNE} seeking with partial information,'' in \emph{The 2022
  American Control Conference (ACC)}, 2022, submitted.

\bibitem{nash1950equilibrium}
J.~F. Nash \emph{et~al.}, ``Equilibrium points in {N}-person games,''
  \emph{Proceedings of the national academy of sciences}, vol.~36, no.~1, pp.
  48--49, 1950.

\bibitem{facchinei2007generalized}
F.~Facchinei, A.~Fischer, and V.~Piccialli, ``On generalized {Nash} games and
  variational inequalities,'' \emph{Operations Research Letters}, vol.~35,
  no.~2, pp. 159--164, 2007.

\bibitem{facchinei2010generalized}
F.~Facchinei and C.~Kanzow, ``Generalized {Nash} equilibrium problems,''
  \emph{Annals of Operations Research}, vol. 175, no.~1, pp. 177--211, 2010.

\bibitem{kannan2011strategic}
A.~Kannan, U.~V. Shanbhag, and H.~M. Kim, ``Strategic behavior in power markets
  under uncertainty,'' \emph{Energy Systems}, vol.~2, no.~2, pp. 115--141,
  2011.

\bibitem{kannan2013addressing}
------, ``Addressing supply-side risk in uncertain power markets: stochastic
  {Nash} models, scalable algorithms and error analysis,'' \emph{Optimization
  Methods and Software}, vol.~28, no.~5, pp. 1095--1138, 2013.

\bibitem{nagurney2020stochastic}
A.~Nagurney, M.~Salarpour, J.~Dong, and L.~S. Nagurney, ``A stochastic disaster
  relief game theory network model,'' in \emph{SN Operations Research Forum},
  vol.~1, no.~2.\hskip 1em plus 0.5em minus 0.4em\relax Springer, 2020, pp.
  1--33.

\bibitem{nikolova2014mean}
E.~Nikolova and N.~E. Stier-Moses, ``A mean-risk model for the traffic
  assignment problem with stochastic travel times,'' \emph{Operations
  Research}, vol.~62, no.~2, pp. 366--382, 2014.

\bibitem{zhou2021robust}
Z.~Zhou, P.~Mertikopoulos, A.~L. Moustakas, N.~Bambos, and P.~Glynn, ``Robust
  power management via learning and game design,'' \emph{Operations Research},
  vol.~69, no.~1, pp. 331--345, 2021.

\bibitem{facchinei2007finite}
F.~Facchinei and J.-S. Pang, \emph{Finite-dimensional variational inequalities
  and complementarity problems}.\hskip 1em plus 0.5em minus 0.4em\relax
  Springer Science \& Business Media, 2007.

\bibitem{koshal2012regularized}
J.~Koshal, A.~Nedic, and U.~V. Shanbhag, ``Regularized iterative stochastic
  approximation methods for stochastic variational inequality problems,''
  \emph{IEEE Transactions on Automatic Control}, vol.~58, no.~3, pp. 594--609,
  2012.

\bibitem{bot2020mini}
R.~Bot, P.~Mertikopoulos, M.~Staudigl, and P.~Vuong, ``Mini-batch
  forward-backward-forward methods for solving stochastic variational
  inequalities,'' \emph{Stochastic Systems}, 2020.

\bibitem{yousefian2017smoothing}
F.~Yousefian, A.~Nedi{\'c}, and U.~V. Shanbhag, ``On smoothing, regularization,
  and averaging in stochastic approximation methods for stochastic variational
  inequality problems,'' \emph{Mathematical Programming}, vol. 165, no.~1, pp.
  391--431, 2017.

\bibitem{cui2016analysis}
S.~Cui and U.~V. Shanbhag, ``On the analysis of reflected gradient and
  splitting methods for monotone stochastic variational inequality problems,''
  in \emph{2016 IEEE 55th Conference on Decision and Control (CDC)}.\hskip 1em
  plus 0.5em minus 0.4em\relax IEEE, 2016, pp. 4510--4515.

\bibitem{kannan2019optimal}
A.~Kannan and U.~V. Shanbhag, ``Optimal stochastic extragradient schemes for
  pseudomonotone stochastic variational inequality problems and their
  variants,'' \emph{Computational Optimization and Applications}, vol.~74,
  no.~3, pp. 779--820, 2019.

\bibitem{salehisadaghiani2016distributed}
F.~Salehisadaghiani and L.~Pavel, ``Distributed {Nash} equilibrium seeking: A
  gossip-based algorithm,'' \emph{Automatica}, vol.~72, pp. 209--216, 2016.

\bibitem{parise2020distributed}
F.~Parise, S.~Grammatico, B.~Gentile, and J.~Lygeros, ``Distributed convergence
  to {Nash} equilibria in network and average aggregative games,''
  \emph{Automatica}, vol. 117, p. 108959, 2020.

\bibitem{yi2019operator}
P.~Yi and L.~Pavel, ``An operator splitting approach for distributed
  generalized {Nash} equilibria computation,'' \emph{Automatica}, vol. 102, pp.
  111--121, 2019.

\bibitem{yi2018distributed}
------, ``Distributed generalized {Nash} equilibria computation of monotone
  games via double-layer preconditioned proximal-point algorithms,'' \emph{IEEE
  Transactions on Control of Network Systems}, vol.~6, no.~1, pp. 299--311,
  2018.

\bibitem{pavel2019distributed}
L.~Pavel, ``Distributed {GNE} seeking under partial-decision information over
  networks via a doubly-augmented operator splitting approach,'' \emph{IEEE
  Transactions on Automatic Control}, vol.~65, no.~4, pp. 1584--1597, 2019.

\bibitem{bianchi2022fast}
M.~Bianchi, G.~Belgioioso, and S.~Grammatico, ``Fast generalized nash
  equilibrium seeking under partial-decision information,'' \emph{Automatica},
  vol. 136, p. 110080, 2022.

\bibitem{belgioioso2020distributed}
G.~Belgioioso, A.~Nedich, and S.~Grammatico, ``Distributed generalized nash
  equilibrium seeking in aggregative games on time-varying networks,''
  \emph{IEEE Transactions on Automatic Control}, 2020.

\bibitem{yang2019survey}
T.~Yang, X.~Yi, J.~Wu, Y.~Yuan, D.~Wu, Z.~Meng, Y.~Hong, H.~Wang, Z.~Lin, and
  K.~H. Johansson, ``A survey of distributed optimization,'' \emph{Annual
  Reviews in Control}, vol.~47, pp. 278--305, 2019.

\bibitem{franci2020distributed}
B.~Franci and S.~Grammatico, ``A distributed forward-backward algorithm for
  stochastic generalized {Nash} equilibrium seeking,'' \emph{IEEE Transactions
  on Automatic Control}, 2020.

\bibitem{cui2021relaxed}
S.~Cui, B.~Franci, S.~Grammatico, U.~V. Shanbhag, and M.~Staudigl, ``A
  relaxed-inertial forward-backward-forward algorithm for stochastic
  generalized {Nash} equilibrium seeking,'' \emph{arXiv preprint
  arXiv:2103.13115}, 2021.

\bibitem{lei2020synchronous}
J.~Lei, U.~V. Shanbhag, J.-S. Pang, and S.~Sen, ``On synchronous, asynchronous,
  and randomized best-response schemes for stochastic {Nash} games,''
  \emph{Mathematics of Operations Research}, vol.~45, no.~1, pp. 157--190,
  2020.

\bibitem{lei2018distributed}
J.~Lei and U.~V. Shanbhag, ``Distributed variable sample-size gradient-response
  and best-response schemes for stochastic nash equilibrium problems over
  graphs,'' \emph{arXiv preprint arXiv:1811.11246}, 2018.

\bibitem{franci2020stochastic}
B.~Franci and S.~Grammatico, ``Stochastic generalized {Nash} equilibrium
  seeking under partial-decision information,'' \emph{arXiv preprint
  arXiv:2011.05357}, 2020.

\bibitem{palomar2010convex}
D.~P. Palomar and Y.~C. Eldar, \emph{Convex optimization in signal processing
  and communications}.\hskip 1em plus 0.5em minus 0.4em\relax Cambridge
  university press, 2010.

\bibitem{ravat2011characterization}
U.~Ravat and U.~V. Shanbhag, ``On the characterization of solution sets of
  smooth and nonsmooth convex stochastic {Nash} games,'' \emph{SIAM Journal on
  Optimization}, vol.~21, no.~3, pp. 1168--1199, 2011.

\bibitem{kulkarni2012variational}
A.~A. Kulkarni and U.~V. Shanbhag, ``On the variational equilibrium as a
  refinement of the generalized {Nash} equilibrium,'' \emph{Automatica},
  vol.~48, no.~1, pp. 45--55, 2012.

\bibitem{huang2021distributed}
Y.~Huang and J.~Hu, ``Distributed solution of {GNEP} over networks via the
  {Douglas-Rachford} splitting method,'' in \emph{2021 IEEE 60th Conference on
  Decision and Control (CDC)}.\hskip 1em plus 0.5em minus 0.4em\relax IEEE,
  2021, to appear, a full version is available at
  https://arxiv.org/abs/2103.09393.

\bibitem{BauschkeHeinzH2017CAaM}
H.~H. Bauschke, \emph{\BIBforeignlanguage{eng}{{Convex Analysis and Monotone
  Operator Theory in {Hilbert} Spaces}}}, 2nd~ed., ser. CMS Books in
  Mathematics, Ouvrages de mathématiques de la SMC, 2017.

\bibitem{bell1965gershgorin}
H.~E. Bell, ``Gershgorin's theorem and the zeros of polynomials,'' \emph{The
  American Mathematical Monthly}, vol.~72, no.~3, pp. 292--295, 1965.

\bibitem{robbins1971convergence}
H.~Robbins and D.~Siegmund, ``A convergence theorem for non-negative almost
  supermartingales and some applications,'' in \emph{Optimizing methods in
  statistics}.\hskip 1em plus 0.5em minus 0.4em\relax Elsevier, 1971, pp.
  233--257.

\bibitem{shor2012minimization}
N.~Z. Shor, \emph{Minimization methods for non-differentiable functions}.\hskip
  1em plus 0.5em minus 0.4em\relax Springer Science \& Business Media, 2012,
  vol.~3.

\bibitem{lacoste2012simpler}
S.~Lacoste-Julien, M.~Schmidt, and F.~Bach, ``A simpler approach to obtaining
  an {O(1/t)} convergence rate for the projected stochastic subgradient
  method,'' \emph{arXiv preprint arXiv:1212.2002}, 2012.

\bibitem{huang2021sdistributed}
Y.~Huang and J.~Hu, ``Distributed computation of stochastic {GNE} with partial
  information: An augmented best-response approach,'' \emph{arXiv preprint
  arXiv:2109.12290}, 2021.

\bibitem{parise2019distributed}
F.~Parise, B.~Gentile, and J.~Lygeros, ``A distributed algorithm for
  almost-nash equilibria of average aggregative games with coupling
  constraints,'' \emph{IEEE Transactions on Control of Network Systems},
  vol.~7, no.~2, pp. 770--782, 2019.

\bibitem{brinkhoff2002framework}
T.~Brinkhoff, ``A framework for generating network-based moving objects,''
  \emph{GeoInformatica}, vol.~6, no.~2, pp. 153--180, 2002.

\bibitem{belgioioso2018projected}
G.~Belgioioso and S.~Grammatico, ``Projected-gradient algorithms for
  generalized equilibrium seeking in aggregative games are preconditioned
  forward-backward methods,'' in \emph{2018 European Control Conference
  (ECC)}.\hskip 1em plus 0.5em minus 0.4em\relax IEEE, 2018, pp. 2188--2193.

\bibitem{dantzig1955linear}
G.~B. Dantzig, ``Linear programming under uncertainty,'' \emph{Management
  science}, vol.~1, no. 3-4, pp. 197--206, 1955.

\bibitem{birge2011introduction}
J.~R. Birge and F.~Louveaux, \emph{Introduction to stochastic
  programming}.\hskip 1em plus 0.5em minus 0.4em\relax Springer Science \&
  Business Media, 2011.

\bibitem{barbarosoglu2004two}
G.~Barbarosoǧlu and Y.~Arda, ``A two-stage stochastic programming framework
  for transportation planning in disaster response,'' \emph{Journal of the
  operational research society}, vol.~55, no.~1, pp. 43--53, 2004.

\bibitem{shapiro2014lectures}
A.~Shapiro, D.~Dentcheva, and A.~Ruszczy{\'n}ski, \emph{Lectures on stochastic
  programming: modeling and theory}.\hskip 1em plus 0.5em minus 0.4em\relax
  SIAM, 2014.

\bibitem{boyd2004convex}
S.~Boyd, S.~P. Boyd, and L.~Vandenberghe, \emph{Convex optimization}.\hskip 1em
  plus 0.5em minus 0.4em\relax Cambridge university press, 2004.

\bibitem{boley2013local}
D.~Boley, ``Local linear convergence of the alternating direction method of
  multipliers on quadratic or linear programs,'' \emph{SIAM Journal on
  Optimization}, vol.~23, no.~4, pp. 2183--2207, 2013.

\bibitem{liang2018local}
J.~Liang, J.~Fadili, and G.~Peyr{\'e}, ``Local linear convergence analysis of
  primal-dual splitting methods,'' \emph{Optimization}, vol.~67, no.~6, pp.
  821--853, 2018.

\bibitem{cui2020variance}
S.~Cui and U.~V. Shanbhag, ``Variance-reduced proximal and splitting schemes
  for monotone stochastic generalized equations,'' \emph{arXiv preprint
  arXiv:2008.11348}, 2020.

\bibitem{asi2019importance}
H.~Asi and J.~C. Duchi, ``The importance of better models in stochastic
  optimization,'' \emph{Proceedings of the National Academy of Sciences}, vol.
  116, no.~46, pp. 22\,924--22\,930, 2019.

\bibitem{polyak1987introduction}
B.~T. Polyak, ``Introduction to optimization. optimization software,''
  \emph{Inc., Publications Division, New York}, vol.~1, 1987.

\end{thebibliography}

\end{document}